\documentclass[10pt,reqno]{amsart}
\RequirePackage[OT1]{fontenc}
\RequirePackage{amsthm,amsmath}
\RequirePackage[numbers]{natbib}
\RequirePackage[colorlinks,linkcolor=blue,citecolor=blue,urlcolor=blue]{hyperref}
\usepackage{amssymb}
\usepackage{anysize}
\usepackage{hyperref}
\usepackage{extarrows}
\usepackage{multirow}
\usepackage{natbib}
\usepackage{stmaryrd}
\usepackage{color}
\usepackage{amsmath}
\usepackage{enumitem}
\usepackage{mathtools} 
\usepackage{dsfont} 
\usepackage{comment}
\usepackage{soul}
\usepackage{braket}
\usepackage{bm}

\DeclareFontFamily{U}{mathx}{}
\DeclareFontShape{U}{mathx}{m}{n}{<-> mathx10}{}
\DeclareSymbolFont{mathx}{U}{mathx}{m}{n}
\DeclareMathAccent{\widecheck}{0}{mathx}{"71}

\numberwithin{equation}{section}
\theoremstyle{plain} 
\newtheorem{theorem}{Theorem}[section]
\newtheorem{lemma}[theorem]{Lemma}

\newtheorem{proposition}[theorem]{Proposition}
\newtheorem{example}[theorem]{Example}
\newtheorem{remark}[theorem]{Remark} 

\theoremstyle{definition}
\newtheorem{definition}[theorem]{Definition}
\newtheorem{assumption}[theorem]{Assumption}

\usepackage{caption} 
\renewcommand{\Re}{\mathrm{Re}\,}
\renewcommand{\Im}{\mathrm{Im}\,}

\newcommand{\E}{{\mathbb E }}

\newcommand{\R}{{\mathbb R }}
\newcommand{\N}{{\mathbb N}}
\newcommand{\Z}{{\mathbb Z}}

\newcommand{\C}{{\mathbb C}}

\newcommand{\m}{{\mathfrak m}}

\newcommand{\s}{{\mathfrak{s}}}

\newcommand{\ii}{\mathrm{i}}
\newcommand{\ee}{\mathrm{e}}

\newcommand{\dd}{\mathrm{d}}

\newcommand{\sgn}{\mathrm{sgn}}

\newcommand{\vertiii}[1]{{\left\vert\kern-0.3ex\left\vert\kern-0.3ex\left\vert #1 
		\right\vert\kern-0.3ex\right\vert\kern-0.3ex\right\vert}}

\newcommand{\nc}{\normalcolor}

\newcommand{\bs}{\boldsymbol}

\def\Tr{\mathrm{Tr}}

\def\<{\langle}
\def\>{\rangle}

\renewcommand{\mathbf}[1]{\bs{#1}}
 
\marginsize{25mm}{25mm}{25mm}{26mm}

\allowdisplaybreaks

\title[Out-of-time-ordered correlators for Wigner matrices]{Out-of-time-ordered correlators for Wigner matrices}
 \author[Cipolloni \and Erd\H{o}s \and Henheik]{}
\begin{document}
	\maketitle
	
	\vspace{0.25cm}

\renewcommand{\thefootnote}{\fnsymbol{footnote}}

\noindent
\mbox{}%
\hfill%
\begin{minipage}{0.21\textwidth}
	\centering
	{Giorgio Cipolloni}\footnotemark[1]\\
	\footnotesize{\textit{gc4233@princeton.edu}}
\end{minipage}
\hfill%
\begin{minipage}{0.21\textwidth}
	\centering
	{L\'aszl\'o Erd\H{o}s}\footnotemark[2]\\
	\footnotesize{\textit{lerdos@ist.ac.at}}
\end{minipage}
\hfill%
\begin{minipage}{0.21\textwidth}
	\centering
	{Joscha Henheik}\footnotemark[2]\\
	\footnotesize{\textit{joscha.henheik@ist.ac.at}}
\end{minipage}
\hfill%
\mbox{}%
\footnotetext[1]{Princeton Center for Theoretical Science and Department of Mathematics, Princeton University, Princeton, NJ 08544, USA.}
\footnotetext[2]{Institute of Science and Technology Austria, Am Campus 1, 3400 Klosterneuburg, Austria.}
		\footnotetext[2]{Supported by ERC Advanced Grant ``RMTBeyond'' No.~101020331.}

\renewcommand*{\thefootnote}{\arabic{footnote}}
\vspace{0.25cm}

 \begin{abstract} {We consider the time evolution of the \emph{out-of-time-ordered correlator} (OTOC) of two general 
 observables $A$ and $B$ in a mean field chaotic 
 quantum system described by a random Wigner matrix as its Hamiltonian.
  We rigorously identify three time regimes separated
 by the physically relevant \emph{scrambling} and \emph{relaxation} times. 
 The main feature of our analysis is that   we express the error
 terms in the optimal Schatten (tracial) norms of the observables, allowing us to track the exact dependence  
 of the errors  on their rank. In particular,  for significantly overlapping observables with low rank
 the OTOC is shown to exhibit a significant local maximum at the scrambling time,
  a feature that may not have been noticed in the physics literature before. Our main tool
  is a novel 
 multi-resolvent local law  with Schatten norms that unifies and improves previous local 
 laws involving  either the much cruder operator norm (cf.~\cite{multiG})  or the Hilbert-Schmidt norm (cf.~\cite{A2}).
 }

\end{abstract}


\section{Introduction} \label{sec:intro}

A basic feature of  a strongly interacting quantum system is that local  initial states become 
non-local along the unitary time evolution,  in particular they become increasingly 
harder to distinguish by local observables. The simplest way to detect this chaotic behavior
is to monitor the overlap $\langle A(t) B \rangle$ of the Heisenberg
time evolution $A(t): = \ee^{-\ii tH}A e^{\ii tH}$ of an observable $A$ with another static observable $B$,
where $H$ is the Hamiltonian and  $A, B$ are Hermitian operators.
Here $\langle M\rangle: = \frac{1}{N} \Tr M$ denotes the normalized trace of 
an $N\times N$ matrix, $N$ is the dimension of the quantum state space.
As time goes on, the overlap between two local observables converges to its stationary value
in a process called
{\it quantum thermalisation}. Since this stationary value is 
practically\footnote{In a closed quantum system with finitely many degrees of freedom the initial state
	is never fully lost as the stationary value still slightly depends on the original overlap $\langle A B\rangle$, but it
	is suppressed by $N$; e.g. it follows from \eqref{ABGUE} that
	$$
	\lim_{t\to\infty} \E_{GUE} \langle A(t) B\rangle  =\Big(1-\frac{1}{N+1}\Big)
	\langle A \rangle\langle B \rangle + \frac{1}{N+1}  \langle A  B \rangle.
	$$
}  factorized, 
$\langle A\rangle\langle B \rangle$,  the original observable $A$ becomes hardly detectable
from its time evolution  by local observables $B$.

A more refined measure of the dynamically evolving quantum chaos is the {\it out-of-time-ordered correlator (OTOC)}
of two observables, 
defined as\footnote{In the physics literature, the OTOC is usually defined without the factor $1/2$. We chose it, however, for convenience.}
\begin{equation}\label{otoc}
	\mathcal{C}_{A,B}(t) := \frac{1}{2}\big\langle \big|[A(t), B]\big|^2 \big\rangle
\end{equation}
measuring the evolution of the commutator of $A(t)$ and $B$. Starting with 
commuting observables, $[A,B]=0$, this quantity initially grows, expressing how the time evolution $A(t)$ of a local observable spreads (or {\it scrambles})  to non-local degrees of freedom
expressed by $B$. 
The moment when this growth stops is called\footnote{We remark that some  papers use slightly 
	different definition, here we  follow the terminology of \cite{2209.07965},
	\cite[Section 3.3]{1706.05400}. } the {\it scrambling time} $t_*$.  Scrambling  is closely related to thermalisation,
but  it typically involves non-local observables $B$.
Thus in 
a quantum system with {\it local}  interactions, the thermalisation time is smaller than $t_*$ and it is independent
of the system size, while the scrambling takes place on a longer time scale until local information is
shared with all degrees of freedom in the system. Beyond the scrambling time,  the OTOC settles to constant value
at a larger time scale called  the {\it relaxation time},  and then it remains essentially unchanged.

The fine distinction between thermalisation and scrambling  became very popular in  physics about $15$ years ago
motivated by  
the fundamental papers by Hayden and Preskill \cite{0708.4025} and Sekino and Susskind \cite{0808.2096} 
related to the black hole information paradox. 
The concept of  the OTOC in quantum chaos research was introduced 
in Kitaev's lectures \cite{Kitaev} on the connection between the 
{\it Sachdev-Ye-Kitaev (SYK) model} and black holes. 
Owing to these fascinating connections, the physics literature on OTOC in various interacting quantum systems 
has become enormous; we refer the reader to the reviews \cite{2202.07060, 2209.07965} and  extensive references therein. 
In contrast, OTOC has basically not been considered in  the mathematical literature apart from \cite{2312.01736}
that studies  a very different model than our current random matrix setup.

Besides the OTOC, quantum chaos has several other signatures: the conventional one is the spectral statistics of the Hamiltonian.
Following E.~Wigner's groundbreaking observation,  in a sufficiently chaotic quantum system
the local eigenvalue statistics are given by the 
universal Wigner-Dyson distribution that depends
only on the basic symmetries of the system. In the physics literature the spectral statistics are often described by the
{\it spectral form factor (SFF)}, or {\it two point spectral correlator}, 
defined as $r_2(t): =  \E |\langle \ee^{\ii tH}\rangle|^2$, where $\E$ indicates a
statistical averaging over an ensemble of Hamiltonians. It is well known that the SFF tends to become universal
for large times\footnote{See the celebrated  {\it slope-dip-ramp-plateau} picture,  e.g.~in \cite{LevandierLombJP}.},
while it still reflects  properties of the actual quantum system (especially its density of states)
for shorter times. A good physics summary is found in~\cite{1706.05400},
while  a recent mathematical analysis of the SFF for general Wigner matrices was
given in \cite{2109.06712}; more precise formulas are available for exactly solvable  ensembles
\cite{2006.00668, 2007.07473, 2206.14950}.

The OTOC is a more refined description of quantum chaos than the SFF, as it also
involves observables. In particular, the SFF misses important features like the sensitivity of chaos to the locality of the observables
or {\it early time chaos}, i.e.~the exponential growth of the OTOC for certain strongly interacting systems like SYK 
(called {\it fast scramblers} \cite{0808.2096}) versus the polynomial growth for {\it slow scramblers} like certain weakly chaotic systems
(see, e.g.~\cite[Section II]{1906.07706}  and references therein).
Note that  the SFF can be recovered from the OTOC by averaging, either over the observables or 
over the unitary group in case of  unitarily invariant Hamiltonians, see \cite{1706.05400}.
For example, if $H$ is a GUE random matrix, then \cite[Eq. (57)-(58)]{1706.05400})
\begin{equation}\label{ABGUE}
	\E_{GUE} \langle A (t)B\rangle  = \langle A \rangle\langle B \rangle + \frac{N^2r_2(t)-1}{N^2-1} 
	\big[\langle A  B \rangle - \langle A \rangle\langle B \rangle\big],
\end{equation}
with 
$$
r_2(t) =  \Big(\frac{J_1(2t)}{t}\Big)^2+\frac{1}{N} - \frac{1}{N} \Big( 1- \frac{t}{2N}\Big)\mathbf {1}(t\le 2N),
$$
where $J_1$ is the Bessel function of the first kind of order one. Thus $r_2(t)$ can be expressed from $\langle A (t)B\rangle$. 
A similar relation holds between the OTOC and the four point spectral correlator.

\medskip

The main goal of the current paper is to give a comprehensive mathematical analysis of the OTOC with general observables $A, B$,
when the Hamiltonian $H$ is a  Wigner matrix. Wigner matrices represent the Hamiltonian of the most chaotic  quantum systems
with matrix elements being independent, identically distributed (i.i.d.) random variables. 
In the physics literature, random matrix theory (RMT) is often used as a test case to see to what extent this relatively simple
model mimics the physics of more complicated systems  such as
interacting many-body models (like SYK) or even  models with nontrivial spatial structure
(like spin chains). Spectral statistics are remarkably robust, especially the universality of the  large time (so-called
{\it plateau}) regime 
of the SFF has proved to be ubiquitous in many different chaotic quantum systems, in accordance with 
the celebrated Bohigas-Giannoni-Schmit 
conjecture \cite{BoGiSch}.
The OTOC is a more delicate quantity and, 
admittedly, its several interesting features that appear 
in more realistic strongly coupled systems are not captured
by RMT.  For example, the early time exponential growth of the OTOC is not present in RMT
and  there is no qualitative difference between the thermalisation and scrambling times  since $H$ is mean field
(see \cite[Section 3.3]{1706.05400} for a detailed analysis). The difference between chaotic and 
integrable systems or the effect of their possible coexistence on the OTOC (studied e.g.~in \cite{1906.07706}) are also not
visible in RMT since Wigner matrices are fully chaotic.   Nevertheless, RMT becomes a good description beyond 
the scrambling time as claimed in \cite[Section 3.3]{1706.05400} and demonstrated in \cite[Section~6]{1706.05400}.
The calculations are performed under the unitary invariance  assumption and without controlling  the error terms.

In the main Theorem~\ref{thm:OTOC} of the current work, 
using very different methods,   we rigorously describe the behaviour of $\mathcal{C}_{A,B}(t)$
up to very long times for general Wigner matrices (no unitary invariance assumed). 
We mimic the  locality of the observables by considering matrices $A, B$ that are far from being full rank and track 
this effect throughout all error terms by using tracial norms that are sensitive to the rank.
We distinguish three time regimes  (see Figure \ref{fig:otoc} and Section~\ref{rmk:interpret}); 
for short times (before the scrambling time $t_*$) we find a quadratic growth in $t$; for intermediate
times we find that $\mathcal{C}_{A,B}(t)$ heavily depends on the ranks of $A, B$ and their overlap $AB$.
In particular, we detect  a remarkable high peak of $\mathcal{C}_{A,B}(t)$ when  the ranks of $A$ and $B$ are 
small but their overlap $\langle AB \rangle$ is still relatively large. 
To our knowledge, this observation may be new even in
the physics literature. We then identify the \emph{relaxation time}, $t_{**}$, when  the
OTOC {\it saturates}, i.e. it
becomes essentially constant (equal its thermal limiting value)
with small oscillations. As expected, in the last regime, after  $t_{**}$, our model behaves universally; a qualitatively similar behaviour 
has been demonstrated for several more complicated systems in the physics literature both theoretically and numerically, see e.g.
\cite{1706.05400, 1906.07706, 2202.07060,  2202.09443, 2209.07965}. 
While for technical reasons we cannot consider infinite times, our analysis is valid for sufficiently long times to see
all physically relevant features. For brevity 
we carry out the proofs  at infinite temperature, but
our methods can easily be extended to any finite temperature and we will give the corresponding formulas
(see Section~\ref{rmk:finite} below).

\begin{figure}[h]
	\centering
	\includegraphics[scale=0.65]{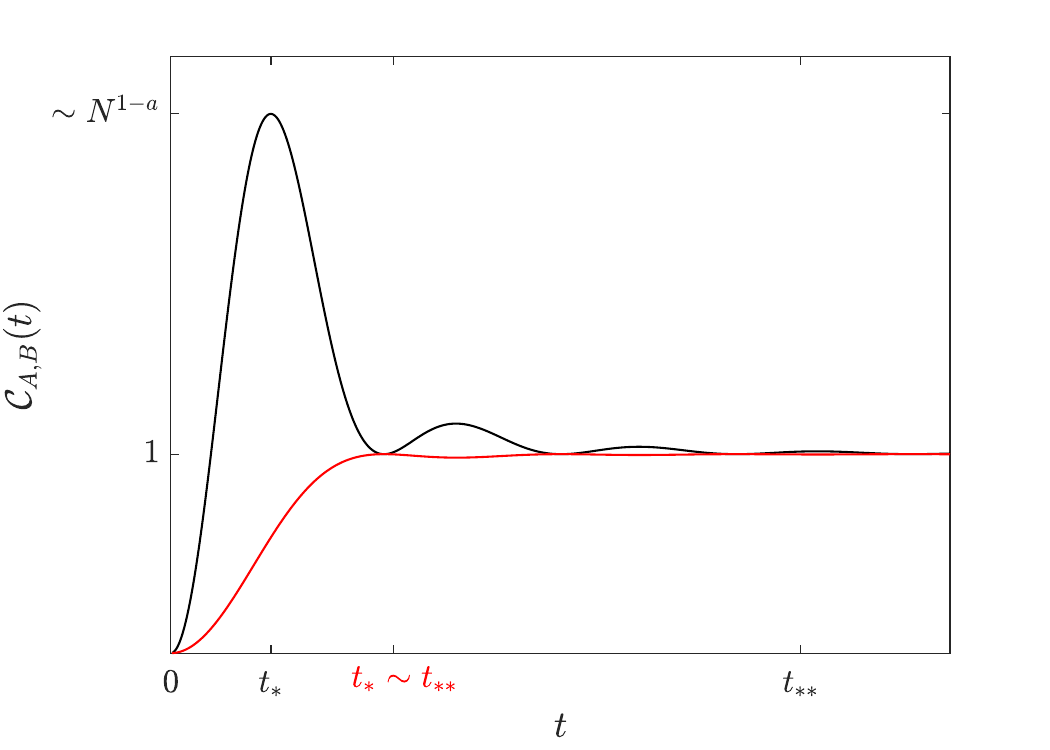}
	\caption{The two curves  show the behaviour of $\mathcal{C}_{A,B}(t)$ in two different scenarios for two
	 commuting traceless observables $A,B$ normalized to $\langle A^2 \rangle = \langle B^2 \rangle  =1$. 
	The black curve represents the case $A = B$ with $\mathrm{rank}(A) = N^{a}$, $a \in [0,1]$
		 where the OTOC exhibits a large peak of size $N^{1-a}$ around the scrambling time $t_* \sim 1$. Afterwards, it decays to its thermal limiting value (normalized to one) around the relaxation time $t_{**} \sim N^{\frac{1-a}{3}}$.  
	The red curve represents the case where $AB = 0$. Here, both  $t_*$ and  $t_{**}$ are of order one,  independent of the ranks of $A$ and $B$. For more details
	see Section~\ref{rmk:interpret}.}
	\label{fig:otoc}
\end{figure}

From  the mathematical point of view, our work is the closest to 
\cite{thermalisation, multiG},  where the deterministic leading term
of traces of products
of observables  at different times, $\langle A_1(t_1) A_2(t_2) A_3(t_3) ....\rangle$, were computed
(see also \cite{multipointCLT, multipointCLTcomb}, where even the Gaussian fluctuations of such chains  were proven).
Clearly the OTOC is a special case of (the difference of two) such chains. The main 
novelty is that now we use only Schatten (tracial) norms\footnote{The (normalized) $p$-th
Schatten norm of a matrix $A \in \C^{N \times N}$ is defined as $\langle |A|^p \rangle^{1/p}$ for $p \in [1, \infty)$,  where $|A|:=(AA^*)^{1/2}$.} of the observables in the estimates, while
\cite{thermalisation, multipointCLT, multipointCLTcomb}  used the much cruder operator norm. 
In particular, we can extend the time scale for the validity of our description.
More importantly, note that the interesting features of the OTOC are manifested for
small rank observables for which the operator norm is a major overestimate and conceptually is an overkill.

The main tool is a  concentration result, called {\it multi-resolvent local law}, for alternating products
 of resolvents
 of random matrices 
and deterministic matrices. More precisely, setting  $G_i:=G(z_i)$ and considering 
deterministic matrices $A_i$,  the main object of interest is
\begin{equation}
	\label{eq:mainobj}
	G_1A_1G_2A_2G_3\dots A_{k-1} G_k
\end{equation}
for some fixed $k$. We will show  
that \eqref{eq:mainobj} concentrates around a deterministic object  and gives an upper bound
on the fluctuation.
The  interesting regime is the local one, i.e. when $|\Im z_i|\ll 1$.
 Resolvents can then be converted to unitary time evolution $\ee^{\ii tH}$ by standard contour integration.
 
Local laws in general assert that resolvents $G(z)$    tend to become
deterministic (with high probability) in the large $N$ limit even if the spectral parameter $z$ is very close to the real axis
(typically  for any $|\Im z|\gg N^{-1}$ in the bulk spectrum). For example, typical {\it single resolvent local laws} for Wigner matrices 
assert that, for any fixed $\xi>0$,\footnote{Traditionally \cite{EYY2012, KnowYin, BEKYY}, 
	local laws did not consider arbitrary test matrix $A$, but only $A=I$ or special rank one projections $A = \bm y \bm x^*$
	leading to the \emph{isotropic local law} in \eqref{eq:singlegllaw}. General $A$ was included later, e.g.~in~\cite{slowcorr}.}
\begin{equation}
\label{eq:singlegllaw}
\big|\langle (G(z)-m(z))A\rangle\big|\le \frac{N^\xi \lVert A\rVert}{N\eta}\,, \qquad \big|\langle {\bm x}, (G(z)-m(z)) {\bm y}\rangle\big|\le \frac{N^\xi \lVert {\bm x}\rVert\lVert {\bm y}\rVert}{\sqrt{N\eta}}\,, \qquad \text{with} \qquad \eta := |\Im z|\,,
\end{equation}
for a deterministic matrix $A\in\C^{N\times N}$ and deterministic vectors ${\bm x},{\bm y}\in \C^N$ with very high probability
as $N$ becomes large. Here, $m(z)$ is the Stieltjes-transform of the Wigner semicircle law:
\begin{equation}
	\label{eq:semicirc}
	m(z)=m_{\mathrm{sc}}(z):=
	\int_\R\frac{1}{(x-z)}\rho_{\mathrm{sc}}(x)\,\dd x, \qquad\quad \rho_{\mathrm{sc}}(x):=\frac{1}{2\pi}\sqrt{[4-x^2]_+}.
\end{equation}
 However, the deterministic limit of a multi-resolvent chain~\eqref{eq:mainobj}
 is not simply $m(z_1) m(z_2) A_1 A_2\dots$, i.e, one cannot mechanically 
replace each $G$ by a scalar $m$, the actual formula is much more complicated, see~\eqref{eq:Mdef} below.
As to the  accuracy of this deterministic approximation,
besides  $N$ and the imaginary part of the spectral parameter, $\eta=\Im z$, the error term 
crucially depends on the appropriate norms of $A_i$
as well as on the distinction whether $A_i$ is traceless or not. 
The fact that traceless observables substantially reduce both the size of the deterministic limit of~\eqref{eq:mainobj} and of its fluctuation
has first been observed and
exploited in \cite{ETHpaper}, see also \cite{multiG} for a comprehensive analysis of arbitrary long chains. The results in
\cite{multiG}  were optimal both in $N$ and $\eta$, but
they all  used  the simplest operator norm of $A_i$'s in the error term which is far from being optimal
for low rank observables.

 Concerning the more accurate norms,
only very recent papers \cite{A2, edgeETH} started deviating from the operator norm in the error terms. 
The main purpose of these papers was to prove the Eigenstate Thermalisation Hypothesis (ETH) for random matrices
(originally posed by Deutsch in \cite{deutsch1991}) in its most optimal form, including  low rank observables. Moreover, 
the key point in \cite{edgeETH} was to obtain ETH also uniformly in the spectrum, including the critical edge regime. This 
required  to focus on a local law for $\langle \Im GA \Im GA\rangle$ and to extract the smallness of order $\rho^2$ 
at the spectral edge owing to the {\it imaginary part} of the resolvents  (here $\rho =\pi^{-1} |\Im m_{\mathrm sc}|$ is the local density of states).
However,  \cite{A2, edgeETH} \emph{exclusively} used the Hilbert-Schmidt (HS) norm, $\langle |A|^2\rangle^{1/2}$,  which caused suboptimal $(N\eta)$-dependence.
For the purpose of  \cite{A2, edgeETH}, this suboptimality in $(N\eta)$ was irrelevant since the proof of ETH relies on local laws in  the $\eta$--regime when
$N\eta$ is practically order one.

In contrast to \cite{A2, edgeETH}, 
for studying the OTOC at shorter times, we need local laws in the regime where $N\eta$ is large
(since  $\eta \sim 1/t$ dictated by the contour integration and $t\ll N$).
In the current paper we prove local laws that are optimal both in $N$ and $\eta$ {\it and} use the 
optimal  Schatten norms of the observables. 
Especially, this allows for a  more accurate description of the
 OTOC in the physically relevant regime of small rank observables. We, however, do not need to
 track the $\rho$-dependence or pay attention to the imaginary parts.  Therefore, 
 the current work and \cite{edgeETH} are complementary;  they effectively handle very different
 aspects of the local law.  While the fundamental idea of these two works is similar, both use
 the {\it Zigzag strategy} described in Section~\ref{sec:proof}, the actual proofs are quite different.
 The main focus in \cite{edgeETH} was to design and handle contour integral representations
 that allowed us to reduce every estimate to resolvent chains involving only $\Im G$'s. 
 In the current paper $\Im G$ plays no role, but we need to track the precise Schatten norms very carefully.  

 To illustrate the strength of our new result in comparison with the previous bounds,  we present the following three estimates 
 for the simplest case $k=2$, with $A_1=A_2=A$, $z_1=z_2$ in the bulk, 
  and ignoring $N^\xi$-factors for some arbitrarily small $\xi > 0$: 
 \begin{equation}\label{3est}
 \Big| \langle GAGA \rangle - m^2 \langle A^2\rangle \Big| \lesssim 
 \begin{cases}
 \dfrac{\|A\|^2}{N\eta}  & \mbox{from \cite[Theorem~2.5]{multiG};}\\
 \sqrt{N\eta} \dfrac{\langle |A|^2 \rangle}{N\eta} & \mbox{from \cite[Theorem~2.2]{A2};  } \\
\dfrac{\langle |A|^2 \rangle}{N\eta}  + \dfrac{ \langle |A|^4\rangle^{1/2}}{N\sqrt{\eta}} & \mbox{from Theorem~\ref{thm:main}.} 
\end{cases}
 \end{equation}
 Note that our current result in the last line of~\eqref{3est}  implies both previous results since
  $\langle |A|^4\rangle^{1/2}\le \sqrt{N} \langle |A|^2\rangle$ and  $\langle |A|^p\rangle^{1/p}\le \|A\|$. While the former bound saturates for low rank observables, the latter saturates for high rank ones. Therefore, our new result with Schatten norms optimally
 interpolates between these two, at least in the bulk regime. 
 
\subsection*{Notations}

By $\lceil x \rceil := \min\{ m \in \Z \colon m \ge x \}$ and $\lfloor x \rfloor := \max\{ m \in \Z \colon m \le x \}$
we denote the upper and lower integer part of a real number $x \in \R$.
For $k \in \N$ we set $[k] := \{1, ... , k\}$, and $\langle A \rangle := d^{-1} \mathrm{Tr}(A)$, $d \in \N$,
for the normalised trace of a $d \times d$-matrix $A$, while $\mathrm{rk} \, A \equiv \mathrm{rank} \, A$ denotes its rank.
For positive quantities $f, g$ we write $f \lesssim g$ resp.~$f \gtrsim g$ and mean that $f \le C g$ resp.~$f \ge c g$ for some $N$-independent constants $c, C > 0$ that may depend only on the basic control parameters $C_p$, see
\eqref{cp}   
in Assumption~\ref{ass:entries} below. Moreover, we will also write $f \sim g$ in case that $f \lesssim g$ and $g \lesssim f$.

We denote vectors by bold-faced lower case Roman letters $\boldsymbol{x}, \boldsymbol{y} \in \C^{N}$, for some $N \in \N$, and define 
\begin{equation*}
	\langle \boldsymbol{x}, \boldsymbol{y} \rangle := \sum_i \bar{x}_i y_i\,, 
	\qquad A_{\boldsymbol{x} \boldsymbol{y}} := \langle \boldsymbol{x}, A \boldsymbol{y} \rangle\,.
\end{equation*}
Matrix entries are indexed by lower case Roman letters $a, b, c , ... ,i,j,k,... $ from the beginning or the middle of the alphabet and unrestricted sums over those are always understood to be over $\{ 1 , ... , N\}$. 

We will use the concept  \emph{'with very high probability'},  meaning that for any fixed $D > 0$, the probability of an $N$-dependent event is bigger than $1 - N^{-D}$ for all $N \ge N_0(D)$. Also, we will use the convention that $\xi > 0$ denotes an arbitrarily small positive exponent, independent of $N$.
Moreover, we introduce the common notion of \emph{stochastic domination} (see, e.g., \cite{semicirclegeneral}): For two families
\begin{equation*}
	X = \left(X^{(N)}(u) \mid N \in \N, u \in U^{(N)}\right) \quad \text{and} \quad Y = \left(Y^{(N)}(u) \mid N \in \N, u \in U^{(N)}\right)
\end{equation*}
of non-negative random variables indexed by $N$, and possibly an additional parameter $u$
from a parameter space $U^{(N)}$, we say that $X$ is stochastically dominated by $Y$, if for all $\epsilon, D >0$ we have 
\begin{equation*}
	\sup_{u \in U^{(N)}} \mathbf{P} \left[X^{(N)}(u) > N^\epsilon Y^{(N)}(u)\right] \le N^{-D}
\end{equation*}
for large enough $N \ge N_0(\epsilon, D)$. In this case we write $X \prec Y$. If for some complex family of random variables we have $\vert X \vert \prec Y$, we also write $X = O_\prec(Y)$. 

\section{Main results} \label{sec:mainres}

We consider $N\times N$ Wigner matrices $W$, i.e.~$W$ is a random real symmetric or complex Hermitian  matrix $W=W^*$ with independent entries (up to the Hermitian symmetry) and with
single entry distributions $w_{aa}\stackrel{\dd}{=}N^{-1/2}\chi_{\dd}$, and $w_{ab}\stackrel{\dd}{=}N^{-1/2}\chi_{\mathrm{od}}$, for $a>b$. The random variables $\chi_{\mathrm{d}},\chi_{\mathrm{od}}$ 
satisfy the following assumptions.\footnote{By inspecting our proof, it is easy to see that  actually we do not need to assume that the off-diagonal entries of $W$ are  identically distributed. We only need that they all have the same second moments, but higher moments can be different.}
\begin{assumption}
\label{ass:entries}
The off-diagonal distribution $\chi_{\mathrm{od}}$ is a real or complex centered random variable, $\E\chi_{\mathrm{od}}=0$, satisfying $\E|\chi_{\mathrm{od}}|^2 = 1$.
The diagonal distribution is a real centered random variable, $\E \chi_{\mathrm{d}} =0$. Furthermore, we assume the existence of high moments, i.e.~for any $p\in \N$ there exists $C_p > 0$ such that
\begin{equation}\label{cp}
\E \big[|\chi_{\mathrm{d}}|^p+|\chi_{\mathrm{od}}|^p\big]\le C_p\,. 
\end{equation}
\end{assumption}

Our main result, Theorem~\ref{thm:OTOC} below, concerns the Heisenberg time evolution
$A(t) := \ee^{-\ii Wt} A \ee^{\ii Wt}$
of a fixed deterministic self-adjoint observable $A = A^* \in \C^{N \times N}$ governed by the Wigner matrix $W$.
More precisely, we consider  the \emph{out-of-time-ordered correlator (OTOC)}
\begin{equation} \label{eq:OTOC}
	\mathcal{C}_{A,B}(t) := \frac{1}{2}\big\langle \big|[A(t), B]\big|^2 \big\rangle =  \big\langle A(t)^2 B^2 \big\rangle - \big\langle A(t) B A(t)B \big\rangle  
\end{equation}
with another self-adjoint observable $B = B^* \in \C^{N \times N}$, 
consisting of a two-point and a four-point\footnote{We remark that some papers (see, e.g., \cite{PappFriPro}) refer to the four-point part $\mathcal{F}_{A,B}(t)$ alone as the OTOC. } part, 
\begin{equation} \label{eq:OTOCparts}
	\mathcal{D}_{A,B}(t)  := \big\langle A(t)^2 B^2 \big\rangle \qquad \text{and} \qquad \mathcal{F}_{A,B}(t)  := \big\langle A(t) B A(t)B \big\rangle\,, \quad \text{respectively.}
\end{equation}

In the formulation of Theorem \ref{thm:OTOC}, a key role is played by the Fourier transform of the semicircular density 
\eqref{eq:semicirc}, 
\begin{equation} \label{eq:J/t}
	\varphi(t) := \widehat{\rho_{\mathrm{sc}}}(t) = \int_{-2}^2 \dd x \rho_{\mathrm{sc}}(x) \ee^{\ii xt} = \frac{J_1(2t)}{t}\,, 
	\quad t\in \mathbb{R}, 
\end{equation}
where we recall that $J_1$ is the Bessel function of the first kind of order one. 
We note that by standard asymptotics of the Bessel functions on the real line, it holds that 
\begin{equation} \label{eq:J1asym}
	J_1(s) = \begin{cases}
- \sgn(s)\cos\left( |s| + \dfrac{\pi}{4} \right) \sqrt{\dfrac{2}{\pi |s|}} + \mathcal{O}\left(\dfrac{1}{|s|^{3/2}}\right)\quad &s \to \pm \infty\,, \\[4mm]
\dfrac{s}{2} - \dfrac{1}{2}\left(\dfrac{s}{2}\right)^3 + \mathcal{O}(|s|^5) \quad &s \to 0 \,. 
	\end{cases} 
\end{equation}

For simplicity, we formulate our main result only for traceless matrices, $\langle A \rangle = \langle B \rangle = 0$ and at infinite temperature. For general observables, see Remark \ref{rmk:additional} and for finite temperature, see Section~\ref{rmk:finite}.
\begin{theorem}[OTOC for Wigner matrices] \label{thm:OTOC}
Let $W$ be a Wigner matrix satisfying Assumption \ref{ass:entries} and let $A, B \in \C^{N \times N}$ be self-adjoint deterministic matrices which are traceless, $\langle A \rangle = \langle B \rangle= 0$. Fix any $\epsilon>0$.
Then, the OTOC \eqref{eq:OTOC} satisfies
\begin{equation} \label{eq:OTOCmain}
	\begin{split}
\mathcal{C}_{A,B}(t) = &\langle A^2 \rangle \langle B^2 \rangle \big[1 - \varphi(t)^2\big] + 2\langle AB \rangle^2 \varphi(t)^2  \, \big[\varphi(2t)    -  \varphi(t)^2 \big] \\[2mm]
&+\langle A^2B^2 \rangle \, \varphi(t)^2 - \langle ABAB \rangle \varphi(t)^4   + \mathcal{O}_\prec \big(\mathcal{E}_{A,B}(t,N)\big)
	\end{split}
\end{equation} 
with an error term 
\begin{equation} \label{eq:errorterm}
\mathcal{E} = \mathcal{E}_{A,B}(t,N) := 
\ee^{t /N^{1/2 - \epsilon}}\left(\frac{|t|^4 \langle A^2\rangle^2}{N} + \frac{|t| \langle A^8 \rangle^{1/2}}{N}\right)^{1/2} \left(\frac{|t|^4 \langle B^2\rangle^2}{N} + \frac{|t| \langle B^8 \rangle^{1/2}}{N}\right)^{1/2}.
\end{equation}
\end{theorem}
The proof of Theorem \ref{thm:OTOC} is given in Section \ref{subsec:OTOCpf} below. It is based on a novel \emph{multi-resolvent local law} with error terms involving optimal  Schatten norms, see Theorem \ref{thm:main}. 

\begin{remark} \label{rmk:additional}We have several comments on Theorem \ref{thm:OTOC}:
	\begin{itemize}
		\item[(i)] \emph{[Non-traceless observables]} For general observables $A,B$, we can decompose them into a tracial and traceless part, $A =: \langle A \rangle + \mathring{A}$, and similarly for $B$. The tracial parts, $\langle A \rangle$ and $\langle B \rangle$, then commute with the unitary time evolution $\ee^{\ii tW}$ and one straightforwardly obtains a result similar to Theorem \ref{thm:OTOC} (see also Remark \ref{rmk:extend}). 
		\item[(ii)] \emph{[Variance of fluctuations]}  The size of the fluctuations around the deterministic leading term in \eqref{eq:OTOCmain}, i.e.~the variance of $\mathcal{C}_{A,B}(t)$, is explicitly computable, following the arguments leading to \cite[Lemma~2.5]{multipointCLT}. The result is expressible purely in terms of Schatten norms of $A$ and $B$ (cf.~\cite[Lemma~2.5 and Definition 3.4]{multipointCLT}), however in \cite{multipointCLT, multipointCLTcomb} the error terms are still in terms of crude operator norms. 
		\item[(iii)] \emph{[Gaussianity]} It is also possible to prove Gaussianity of the fluctuations of $\mathcal{C}_{A,B}(t)$ (cf.~\cite[Theorem~2.7 and Corollary 2.12]{multipointCLT}) by showing an approximate Wick theorem for resolvent chains, similarly to \cite[Theorem~3.6]{multipointCLT}. However, we refrain from doing this for brevity of the current paper. 
	\end{itemize}
\end{remark}

\subsection{Physical interpretation of Theorem \ref{thm:OTOC} by two examples}\label{rmk:interpret}
We will now discuss the behavior of $\mathcal{C}_{A,B}(t)$ in two exemplary and extreme situations of observables $A,B$.
In the first example we will set the two observables identical, in the second we will assume that their product vanishes.
More concretely,  we define 
	\begin{equation} \label{eq:ABdefex}
		\begin{split}
A &= N^{\frac{1-a}{2}} \mathrm{diag}(1, -1, ..., 1, -1, 0,..., 0)\,, \qquad a \in [0,1)\,, \\
B&= N^{\frac{1-b}{2}} \mathrm{diag}(0,...,0,1,-1, ... , 1,-1)\,, \qquad b \in [0,1)\,,
		\end{split}
	\end{equation}
	in such a way that $AB = BA = 0$, and $\langle A \rangle = \langle B \rangle = 0$ as well as $\langle A^2 \rangle = \langle B^2 \rangle = 1$, i.e.~$A$ (resp.~$B$) contains $N^a$-many (resp.~$N^b$-many) non-zero entries on the diagonal. 
	
	\medskip 
	\noindent{\bf Example 1.} For the first example, we have
	\begin{equation} \label{eq:1ex}
	\mathcal{C}_{A,A}(t) = \langle A^2 \rangle^2  \big[1 - \varphi(t)^2\big\{ 1 - 2\varphi(2t) +2\varphi(t)^2  \}\big]  +\langle A^4 \rangle \, \big[\varphi(t)^2 -  \varphi(t)^4 \big]  + \mathcal{O}_\prec \big(\mathcal{E}_{A,A}(t,N)\big) \,.
	\end{equation}
\medskip
	{\bf Example 2.} For the second example, we have
	 \begin{equation} \label{eq:2ex}
	\mathcal{C}_{A,B}(t) = \langle A^2 \rangle \langle B^2 \rangle \big[1 - \varphi(t)^2\big]   + \mathcal{O}_\prec \big(\mathcal{E}_{A,B}(t,N)\big) \,. 
	 \end{equation}
	 The key features of these two examples \eqref{eq:1ex}--\eqref{eq:2ex} are briefly summarized in Table \ref{tab:1}. Ignoring the respective error terms, $\mathcal{C}_{A,A}(t)$ and $\mathcal{C}_{A,B}(t)$ are schematically depicted in Figure \ref{fig:otoc}. 
	 Now we comment on each time regime.
	\begin{itemize}
\item[(i)] [{\bf Short-time regime}] By the asymptotics in \eqref{eq:J1asym}, we have the short-time asymptotic
\begin{alignat*}{2}
	\mathcal{C}_{A,A}(t) &= t^2  \, \big(\langle A^4 \rangle - \langle A^2 \rangle^2\big) + \mathcal{O}\big(|t|^4 \langle A^4 \rangle\big) + \mathcal{O}_\prec \big(\mathcal{E}\big)   \qquad \qquad &&\mbox{(Example 1)}\,,  \\[2mm]
\mathcal{C}_{A,B}(t) &= t^2  \, \langle A^2 \rangle \langle B^2 \rangle + \mathcal{O}\big(|t|^4 \langle A^2 \rangle \langle B^2 \rangle \big) + \mathcal{O}_\prec \big(\mathcal{E}\big)  \qquad \qquad && \mbox{(Example 2)} \,. 
 \end{alignat*}
This shows, that the OTOC for Wigner matrices does \emph{not} exhibit the exponential increase $\sim \ee^{2 \lambda t}$, observed for quantum systems with a \emph{classically} chaotic analogue, where $\lambda$ is the Lyapunov exponent of the classical system. Instead, the OTOC \eqref{eq:OTOC} behaves polynomially, as expected for quantum chaotic systems without a classical analogue 
(e.g.~for certain spin$-\tfrac{1}{2}$ chains \cite{1906.07706}). 

\item[(ii)] [{\bf Scrambling time}]  The monotonous growth of both $\mathcal{C}_{A, A}(t)$ and $\mathcal{C}_{A,B}(t)$ stops at a time of order one (using elementary properties of $\varphi$ from \eqref{eq:J/t}), hence 
 the scrambling time is $t_* \sim 1$. However, the maximally attained value strongly differs for the two examples: While $\mathcal{C}_{A,A}(t_*) \sim N^{1-a}$ heavily depends on $a$ (i.e.~the rank of $A$), the peak of $\mathcal{C}_{A,B}(t_*) \sim 1$ is independent of the ranks of $A$ and $B$. 

\item[(iii)] [{\bf Intermediate regime up to the relaxation time}] The following interval of intermediate times is characterized by a decay of the OTOC \eqref{eq:OTOCmain} towards its thermal limiting value $\langle A^2 \rangle \langle B^2 \rangle = 1$ up to the 
relaxation time $t_{**}$. This regime is also quite different for the two examples \eqref{eq:1ex}--\eqref{eq:2ex}: While for Example 1
 the interval of intermediate times is given by $t\in [t_*, t_{**}]$ where
$t_{**} \sim N^{(1-a)/3}$, in Example 2 the relaxation time $t_{**} \sim 1$ is comparable with the
 scrambling time. 
However,  for technical reasons, the \emph{entire} interval of intermediate times is only accessible if 
$a > 5/11$ for Example 1, and $a+b > 2/3$ for Example 2, otherwise  the leading terms in \eqref{eq:1ex}--\eqref{eq:2ex} 
become smaller than their respective error term. 
In comparison, computing the OTOC with the operator norm in the error terms  \cite[Corollary~2.7]{multiG}, would lead to the (more restrictive) conditions $a > 5/8$ for Example 1, and $a+b>1$ for Example~2.

\item[(iv)] [{\bf Long-time regime}]  In the consecutive long-time regime, i.e.~$t \gg N^{(1-a)/3}$ for  Example 1 
and $t \gg 1$ for Example 2, we find the OTOC \eqref{eq:OTOCmain} 
to concentrate around its thermal limiting value $\langle A^2 \rangle \langle B^2 \rangle  =1$ with small oscillations. 
This confirms the expectation, that the OTOC in strongly chaotic systems exhibits only small fluctuations for long times. These are accessible up to 
\begin{equation*}
t \le N^{\min\left\{ \frac{1}{4}, \, \frac{3a-1}{2} \right\}-\epsilon} \qquad \text{and} \qquad t \le N^{\min \left\{ \frac{1}{4}, \, \frac{3a+1}{10}, \, \frac{3b+1}{10}, \,  \frac{3(a+b)-2}{4} \right\}-\epsilon}
\end{equation*} 
for Example 1 and 2, respectively. 
Again, in comparison with Theorem \ref{thm:OTOC}, the operator norm error terms from \cite[Cor.~2.7]{multiG} would lead to the constraints $t \le N^{\frac{2a-1}{2}-\epsilon}$ for Example 1, and $t \le N^{\frac{a+b-1}{2}-\epsilon}$ for Example~2. 
	\end{itemize}

	\begin{table}[h]
	\begin{tabular}{| c | c | c |}
		\hline
		& Ex. 1\,~\eqref{eq:1ex}: 
		$A = B$, $\mathrm{rk} \, A= N^a$ 
		& Ex. 2\,~\eqref{eq:2ex}: $AB = 0$, $\mathrm{rk} \, A = N^a$, $\mathrm{rk} \, B = N^b$ \\ \hline 
		short times & $\mathcal{C}_{A,A}(t) \sim t^2  \, \big(\langle A^4 \rangle - \langle A^2 \rangle^2\big)$  & $\mathcal{C}_{A,B}(t) \sim  t^2  \, \langle A^2 \rangle \langle B^2 \rangle$ \\ \hline
		scrambling time & $t_* \sim 1$ and $\mathcal{C}_{A,A}(t_*) \sim N^{1-a}$ & $t_* \sim 1$ and $\mathcal{C}_{A,B}(t_*)  \sim 1$ \\ \hline 
		intermediate times & full access if $a > 5/11$  
		& full access if $a+b >2/3$ \\ \hline
		relaxation time & $t_{**} \sim N^{\frac{1-a}{3}}$ and $\mathcal{C}_{A,A}(t_{**}) \sim \langle A^2 \rangle^2$ & $t_{**} \sim 1$ and $\mathcal{C}_{A,B}(t_{**}) \sim \langle A^2 \rangle \langle B^2 \rangle$ \\ \hline
		long times & up to 
		$t \le N^{\min\left\{ \frac{1}{4}, \, \frac{3a-1}{2} \right\}-\epsilon}$ 
		& up to  		$t \le N^{\min \left\{ \frac{1}{4}, \, \frac{3a+1}{10}, \, \frac{3b+1}{10}, \,  \frac{3(a+b)-2}{4} \right\}-\epsilon}$ \\ \hline
	\end{tabular}
	\caption{Overview of the two examples~\eqref{eq:1ex} and~\eqref{eq:2ex}.}
	\label{tab:1}
\end{table}

\subsection{Finite temperature case}\label{rmk:finite}

Theorem \ref{thm:OTOC} can easily be extended to the case of finite temperature,  
$\beta = 1/T>0$. The OTOC \eqref{eq:OTOC} now is given by
\begin{equation} \label{eq:OTOCbeta}
	\mathcal{C}^{(\beta)}_{A,B}(t) := \frac{1}{2} \frac{\Tr\big[  |[A(t),B]|^2 \ee^{-\beta W}\big]}{Z}
\end{equation}
with partition function $Z = \Tr \big[\ee^{-\beta W}\big]$. In this case, the analog of \eqref{eq:OTOCmain} in the regime\footnote{We restrict to this regime for simplicity, as all the error terms can be absorbed into $\prec$.} $\beta \ll \log N$ reads
\begin{equation} \label{eq:CABbeta}
	\begin{split}
		\mathcal{C}^{(\beta)}_{A,B}(t) = &\langle A^2 \rangle \langle B^2 \rangle \big[1 - \varphi(t)^2\big] + \langle AB \rangle^2 \frac{\varphi(t) }{\varphi(\ii \beta)} \, \Re \bigg[\varphi(2t)  \varphi(t + \ii \beta) + \varphi(t) \varphi(2t+ \ii \beta) - 2 \varphi(t)^2 \varphi(t+\ii \beta)\bigg] \\[2mm]
		&+\langle A^2B^2 \rangle \, \varphi(t)^2 - \langle ABAB \rangle \frac{\varphi(t)^3 \Re [\varphi(t+ \ii \beta)]}{\varphi(\ii \beta)} + \mathcal{O}_\prec \left( \mathcal{E}_{A,B}(t,N) \right)\,,
	\end{split}
\end{equation}
where $\mathcal{E}$ is from \eqref{eq:errorterm}. Here $\varphi(z)$ is the complex extension of
$\varphi(t)$ for $z \in \C$; note that $\varphi(z)$
is generically complex but $\varphi(\ii \beta)$ is real.

\begin{figure}[h]
	\centering
	\includegraphics[scale=0.65]{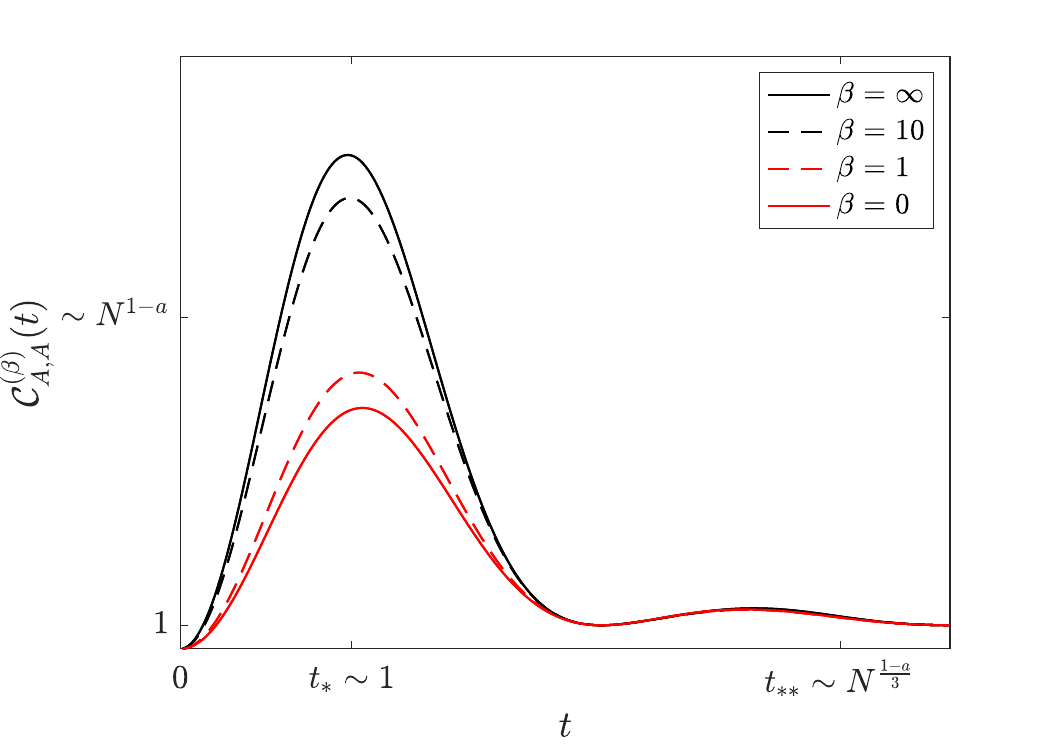}
	\caption{Depicted are four curves illustrating the influence of $\beta = 1/T$ on the OTOC $\mathcal{C}_{A,A}^{(\beta)}(t)$ up to intermediate times for Example 1 from Section \ref{rmk:interpret} (i.e.~normalized to $\langle A^2 \rangle = 1$ with $\mathrm{rank} \, A = N^{\frac{1-a}{2}}$ and $a \in [0,1]$). As $\beta$ increases, the characteristic rank-dependent peak of size $\sim N^{1-a}$ around the scrambling time $t_* \sim 1$ becomes more pronounced and very slightly shifted to the left. }
	\label{fig:peak}
\end{figure}

Moreover, using the asymptotics of the Bessel function in the complex plane, we have  
$$
\varphi(z) = -\sqrt{\frac{1}{\pi z^3}} \left( \cos\left(2z+\frac{\pi}{4}\right)  + \ee^{2|\Im z|} \mathcal{O}\left(|z|^{-1}\right)\right) \qquad \text{for} \qquad |z| \to \infty \quad \text{with} \quad |\arg z| < \pi\,.
$$
In particular,  the thermal limiting value of $\mathcal{C}^{(\beta)}_{A,B}(t)$ is \emph{independent} of $\beta$
at least in our regime $\beta\ll \log N$. Note that for much  larger $\beta\gtrsim \sqrt{N}$ physics calculations
predict  a temperature dependence of the thermal limiting value of the OTOC , see \cite[Eqs.~(3.8)--(3.9)]{2202.09443}.

However, before the long-time regime and neglecting the error term in \eqref{eq:CABbeta}, we find a strong dependence of the OTOC on temperature for Example~1 from Section \ref{rmk:interpret} (cf.~\eqref{eq:ABdefex}--\eqref{eq:1ex}) as illustrated in Figure \ref{fig:peak}. 
In contrast to that, for Example~2 from Section \ref{rmk:interpret} (cf.~\eqref{eq:ABdefex} and \eqref{eq:2ex}), the whole OTOC curve (as depicted in Figure \ref{fig:otoc}) is independent of temperature at any time. This is because all the $\beta$ dependent terms in \eqref{eq:CABbeta} drop out since $AB = BA = 0$.

\section{Proof of Theorem \ref{thm:OTOC}: Multi-resolvent local law with Schatten norms} \label{sec:mainSch}
Theorem~\ref{thm:OTOC} relies on  a new multi-resolvent local law for
alternating chains \eqref{eq:mainobj}  with  deterministic matrices via simple contour integration (see Section \ref{subsec:OTOCpf}).
Its main novelty, using Schatten norms for the observables and still keeping 
optimality in $N$ and $\eta$
, has already been explained in the Introduction.
After collecting some preliminary  information, in Section~\ref{sec:multi} we
 present our new local law, Theorem~\ref{thm:main}, then in Section~\ref{subsec:OTOCpf} we quickly complete the proof
 of Theorem~\ref{thm:OTOC}. Starting from Section~\ref{sec:proof} we will focus 
 on the proof of Theorem~\ref{thm:main}.

\subsection{Preliminaries on the deterministic approximation} \label{subsubsec:prelimM}
Before stating our main technical result, we introduce some additional notation. Given a  non-crossing partition $\pi$ of 
the set $[k]:=\set{1,\ldots,k}$ arranged in cyclic order,  the partial trace $\mathrm{pTr}_{\pi}$ is defined as
\begin{equation}
\label{eq:partrdef}
    \mathrm{pTr}_\pi(B_1,\ldots,B_{k-1}): = \prod_{S\in\pi\setminus \mathfrak{B}(k)}\left\langle\prod_{j\in S}B_j\right\rangle\prod_{j\in \mathfrak{B}(k)\setminus\set{k}} B_j,
\end{equation} 
with $\mathfrak{B}(k)\in\pi$ denoting the unique block containing $k$. Then, for generic $B_i$'s, the deterministic approximation of \eqref{eq:mainobj} is given by (see \cite[Theorem~3.4]{thermalisation})
\begin{equation}
\label{eq:Mdef}
    M_{[1,k]} = M(z_1,B_1,\ldots,B_{k-1},z_k) := \sum_{\pi\in\mathrm{NC}([k])}\mathrm{pTr}_{K(\pi)}(B_1,\ldots,B_{k-1}) \prod_{S\in\pi} m_\circ[S ],
\end{equation}
where $\mathrm{NC}([k])$ denotes the non-crossing partitions of the set $[k]$, and $K(\pi)$ denotes the Kreweras complement of $\pi$ (see \cite[Definition 2.4]{thermalisation} and \cite{Kreweras}). Furthermore, for any subset $S\subset [k]$ we define $m[S]:=m_\mathrm{sc}[\bm{z}_S]$ as the iterated divided difference of $m_\mathrm{sc}$ evaluated in $\bm{z}_S:=\{z_i: i\in S\}$, and by $m_\circ[\cdot ]$ denote the free-cumulant transform of $m[\cdot]$ which is uniquely defined implicitly by the relation 
\begin{equation}
    \label{eq:freecumulant}
    m[S] = \sum_{\pi\in\mathrm{NC}(S)} \prod_{S'\in\pi} m_\circ[S'],  \qquad \forall S\subset [k],
\end{equation}
e.g. $m_\circ[i,j]=m[\set{i,j}]-m[\set{i}]m[\set{j}]$. Throughout the paper, we will often use the fact that $m_\mathrm{sc}[\bm{z}_S]$ can be written as follows
\begin{equation}\label{msc dd}
    m_\mathrm{sc}[\set{z_i: i \in S}] = \int_{-2}^2\rho_\mathrm{sc}(x)\prod_{i\in S}\frac{1}{(x-z_i)}\dd x.
\end{equation}

In order to formulate bounds on the deterministic approximation $M$ as well as the local law bounds in a concise way, we introduce the following {\it weighted Schatten norms.}
\begin{definition}[$\ell$-weighted Schatten norms] \label{def:plSchatten}
For $B \in \C^{N \times N}$, $\ell > 0$ and $p \in [2, \infty]$, we define the \emph{$\ell$-weighted $(2,p)$-Schatten norm} as
\begin{equation} \label{eq:plSchatten}
\vertiii{B}_{p,\ell} := \frac{\langle |B|^2 \rangle^{1/2}}{\ell^{1/2}} + \frac{\langle |B|^p\rangle^{1/p}}{\ell^{1/p}} \quad \text{for} \quad p<\infty \qquad \text{and} \qquad \vertiii{B}_{\infty,\ell} := \frac{\langle |B|^2 \rangle^{1/2}}{\ell^{1/2}} + \Vert B \Vert\,. 
\end{equation}
By elementary inequalities, we have 
\begin{equation} \label{eq:plSchattenrel}
\vertiii{B}_{p, \ell} \lesssim \vertiii{B}_{q, \ell} \qquad \text{and} \qquad \vertiii{B}_{q, \ell} \lesssim \left[1 + (N \ell)^{\frac{q-p}{pq}}\right] \, \vertiii{B}_{p, \ell}\,, \quad \mbox{$2 \le p \le q \le \infty$.}
\end{equation}
\end{definition}

The next lemma gives a bound on the deterministic approximation $M_{[1,k]}$ with traceless observables; 
its proof is given in Appendix \ref{sec:addtech}. We will henceforth follow the convention that the letter $B$ is used for generic matrices, while $A$ is reserved for traceless ones. 

\begin{lemma}[$M$-bound] \label{lem:Mbound}
Fix $k \ge 1$. Consider spectral parameters $z_1, ... , z_{k} \in \C \setminus \R$ and traceless matrices $A_1, ... , A_k \in \C^{N \times N}$.  Denoting $\ell := \min_{j \in [k]} \eta_j \rho_j$ with $\eta_j := |\Im z_j|$ and $\rho_j := \pi^{-1} |\Im m(z_j)|$, it holds that\footnote{\label{ftn:odd}We point out that the bound can be improved to $\ell \prod_{i \in [k]}\left(\ell^{(1-k)/(2k)}\langle |A_i|^3 \rangle^{1/k}\langle |A_i|^2 \rangle^{(k-3)/(2k)}  + \ell^{-1/k} \langle |A_i|^k \rangle^{1/k}\right)$ in 
case of odd $k \ge 3$, but we do not follow this improvement for brevity and ease of notation.}
 \begin{equation} \label{eq:Mbound}
	\left| \langle {M}_{[1,k]}(z_1, A_1, ... ,  z_k) A_k \rangle \right| \lesssim \mathbf{1}(k\ge 2) \ \ell \prod_{i \in [k]} \vertiii{A_i}_{k,\ell} \,,
\end{equation}
where $\vertiii{A_i}_{k,\ell}$ was defined in \eqref{eq:plSchatten}. 
\end{lemma}

It is easy to see that the bound \eqref{eq:Mbound} is optimal for $k$ even (cf.~Footnote \ref{ftn:odd}) and in case that $A_i = A$ for all $i \in [k]$ and $z_i \in \{z, \bar{z}\}$ for some $z \in \C\setminus \R$ in the \emph{bulk}, i.e.~$\Re z \in [-2+\delta, 2-\delta]$ for some $\delta > 0$. 

\subsection{Multi-resolvent local law}\label{sec:multi}
We now formulate our main technical result, an \emph{averaged multi-resolvent local law with Schatten norms},
in Theorem \ref{thm:main}; its proof is given in Section \ref{sec:proof}. 
The corresponding \emph{isotropic multi-resolvent local law} will be formulated in Theorem \ref{thm:isolaw} later.

\begin{theorem}[Averaged multi-resolvent local law with Schatten norms] \label{thm:main}
	Let $W$ be a Wigner matrix satisfying Assumption~\ref{ass:entries}, and fix $k \in \N$. Consider spectral parameters $z_1, ... , z_{k} \in \C \setminus \R$, the associated resolvents $G_j = G(z_j) := (W-z_j)^{-1}$, and traceless deterministic matrices $A_1, ... , A_k \in \C^{N \times N}$. 	Denote 
	$$\eta_j :=  |\Im z_j| \,, \qquad \rho_j := \pi^{-1} |\Im m(z_j)|\,, \qquad \ell := \min_{j \in [k]} \eta_j \rho_j,$$
		and let ${M}_{[1,k]}$ be given in \eqref{eq:Mdef}. Then, for any fixed $\epsilon > 0$ and recalling \eqref{eq:plSchatten},	we have 
\begin{equation} \label{eq:mainresult}
	\left| \langle {G}_1 A_1 ... {G}_k A_k \rangle - \langle {M}_{[1,k]}A_k \rangle \right| \prec \frac{1}{N}\prod_{i \in [k]} \vertiii{A_i}_{2k, \ell}\,,
\end{equation}
uniformly in spectral parameters satisfying $ N \ell \ge N^{\epsilon}$ and $\max_j |z_j| \le N^{1/\epsilon}$. 
\end{theorem}
Notice that, in the bulk regime, $\ell$ essentially agrees with $\eta := \min_{i \in [k]} |\Im z_i|$, since $\rho_i \sim 1$. However, for \eqref{eq:mainresult} to be valid in any regime, the standard condition $N \eta \gg 1$ for the local law in the bulk needs to be replaced by $ N \ell \gg 1$. This ensures that we are at the \emph{mesoscopic scale}, i.e.~there are many eigenvalues in a local window of size $\eta_i$ around each $\Re z_i$.

As already mentioned in the Introduction around \eqref{3est}, Theorem~\ref{thm:main} unifies and improves 
the previous local laws (in the bulk spectrum) with an error term involving either the
 operator norm (see \cite[Eq.~(2.11a) in Theorem 2.5]{multiG}) or Hilbert-Schmidt norm (see \cite[Theorem~2.2]{A2}).
 This follows by estimating (in the relevant $\ell \lesssim 1$ regime)
\begin{equation} \label{eq:SchattenNormHSav}
\vertiii{A_i}_{2k,\ell} \lesssim \begin{cases}
\dfrac{ \Vert A_i \Vert}{\ell^{1/2}} \qquad &\text{for \  \cite[Eq.~(2.11a) in Theorem 2.5]{multiG}} \\[2mm]
(N\ell)^{\frac{k-1}{2k}}  \dfrac{\langle |A_i|^2 \rangle^{1/2}}{\ell^{1/2}} \qquad & \text{for \ \cite[Theorem~2.2]{A2}}
\end{cases}
\end{equation}
in \eqref{eq:mainresult} for $N \ell > 1$ and every $i \in [k]$ by means of elementary inequalities. 

 \begin{remark}[Optimality] \label{rmk:opt}
	The bound \eqref{eq:mainresult} is optimal in case that $A_i = A$ for all $i \in [k]$ and $z_i \in \{z, \bar{z}\}$ for some $z \in \C\setminus \R$ with $|\Im z| \ge N^{-1+\epsilon}$ in the bulk. For $W$ being GUE, this can easily be checked by spectral decomposition of the resolvents and using Weingarten calculus \cite{Weingarten} for the Haar-distributed eigenvectors. The so-called ``ladder diagram" gives the (first) Hilbert-Schmidt term in the estimate
	\begin{equation} \label{eq:Weingarten}
\left| \langle (GA)^k \rangle - \langle  M_{[1,k]}A \rangle \right| \prec \frac{\vertiii{A}_{2k, \ell}^k}{N} \sim \frac{\langle |A|^2\rangle^{k/2}}{N \ell^{k/2}} + \frac{\langle |A|^{2k}\rangle^{1/2}}{N \ell^{1/2}}\,. 
	\end{equation}
Interestingly, in some regimes the second term in \eqref{eq:Weingarten} (non ladder diagram) gives the main contribution, defying the general belief that always the ladder diagrams are the leading terms.
\end{remark}

\begin{remark}[Extensions] \label{rmk:extend}
In Theorem \ref{thm:main}, each $G$ may also be replaced by a product of $G$'s or even $|G|$'s (absolute value). We refrain from stating these results explicitly, as they can be easily obtained from appropriate integral representations (see \eqref{eq:intrep} and \eqref{eq:absGintrep} below). We formulate only the following example for $k=2$ and identical observables $B_1=B_2$ for illustration.
 Let $B \in \C^{N \times N}$ be an arbitrary (i.e.~not necessarily traceless) deterministic matrix. Then, decomposing $B = \mathring{B} + \langle B \rangle$, we have
\begin{equation} \label{eq:avexample}
	\begin{split}
		\langle G_1BG_2B\rangle&=\langle B\rangle^2\langle G_1G_2\rangle+\langle B\rangle\langle G_1G_2\mathring{B}\rangle+\langle B\rangle\langle G_2G_1\mathring{B}\rangle+\langle G_1\mathring{B}_1G_2\mathring{B}_2\rangle \\
		&= \frac{m_1m_2 \langle B\rangle^2}{1-m_1m_2}+m_1m_2\langle \mathring{B}^2\rangle +\mathcal{O}_\prec\left(\frac{|\langle B\rangle|^2}{N\ell^2}+\frac{|\langle B\rangle|\langle|\mathring{B}|^2\rangle^{1/2}}{N\ell^{3/2}}+\frac{\langle|\mathring{B}|^2\rangle}{N\ell}+\frac{\langle|\mathring{B}|^4\rangle^{1/2}}{N\ell^{1/2}}\right)\,. 
	\end{split}
\end{equation}
The statement for different observables $B_1, B_2$ is analogous. 
 \end{remark}

 \subsection{Out-of-time-ordered correlators: Proof of Theorem~\ref{thm:OTOC}} \label{subsec:OTOCpf}
 In order to prove Theorem \ref{thm:OTOC}, we distinguish the following three time regimes, 
 \begin{equation} \label{eq:regimes}
\mathrm{(i)} \  |t|<1 \,,\qquad \mathrm{(ii)} \ 1 \le |t|\le N^{(1-\epsilon)/2}\,, \quad \text{and} \quad \mathrm{(iii)} \  |t| >N^{(1-\epsilon)/2}
 \end{equation}
 for some small fixed $\epsilon > 0$ from Theorem \ref{thm:OTOC}. In the following, we focus on 
 the most complicated case (ii) and discuss the other two cases briefly at the end of this section. 
 Since the arguments in this section are fairly standard, we will leave some irrelevant technical details to the reader. 

 For $1 \le |t| \le N^{(1-\epsilon)/2}$, we employ the integral representation
\begin{equation} \label{eq:eitWrep}
\ee^{\pm \ii t W} = \frac{1}{2 \pi \ii} \oint_{\Gamma} \ee^{\pm \ii tz} G(z) \dd z
\end{equation}
with  the contour
\begin{equation} \label{eq:contourOTOC}
\Gamma \equiv \Gamma_{t, R} := \partial \big( [-R,R] \times \ii [-|t|^{-1}, |t|^{-1}] \big)
\end{equation}
parametrized counterclockwise and the parameter $R$ chosen as $R = N^\kappa$ for some arbitrarily small (but fixed) $\kappa > 0$. In this way, we can write the four-point part $\mathcal{F}_{A,B}(t)$ in \eqref{eq:OTOCparts} of the OTOC \eqref{eq:OTOC} as 
\begin{equation} \label{eq:fourpointint}
\mathcal{F}_{A,B}(t) = \left(\prod_{i \in [4]}\oint_{\Gamma} \frac{\dd z_i}{2 \pi \ii} \right) \ee^{\ii t(z_1 + z_3 - z_2 - z_4)} \langle G(z_1) A G(z_2) B G(z_3) A G(z_4)B \rangle \,. 
\end{equation}

For the part, where \emph{all} $z_i$'s, $i \in [4]$, run on the horizontal parts of the contour $\Gamma$, we can replace $\langle G(z_1) A ... G(z_4)B \rangle$ by $ \langle M(z_1, A, ... , z_4)B \rangle$ at the price of an error $\mathcal{O}_\prec(\mathcal{E})$ with $\mathcal{E} = \mathcal{E}_{A,B}(t)$ from \eqref{eq:errorterm} by means of Theorem \ref{thm:main} for $k=4$. Here, we used $\ell \gtrsim N^{-2 \kappa}|t|^{-2}$, which follows by $\rho(z_i) \gtrsim |\Im z_i|(1+\mathrm{dist}(z_i, [-2,2])^2)^{-1}$, and that $\kappa > 0$ is arbitrarily small and hence $N^{\kappa}$ can be absorbed into $\prec$. 

If \emph{all} the $z_i$'s run on the vertical parts of $\Gamma$, we employ the global law version of our main result, Proposition \ref{prop:initial} below, to replace $\langle G(z_1) A ... G(z_4)B \rangle$ by $ \langle M(z_1, A, ... , z_4)B \rangle$, now at the price of an error $\mathcal{O}_\prec\big(\langle A^8 \rangle^{1/4} \langle B^8 \rangle^{1/4}/N\big)$, which can easily be included in $\mathcal{O}_\prec(\mathcal{E})$. We point out that Theorem \ref{thm:main} cannot be used in this regime, since $\ell = 0$ when one of the $z_i$'s crosses the real axis. 

In the remaining cases, where some of the $z_i$'s run on the horizontal parts and others run on the vertical parts, we can no longer cleanly apply either the local or global law (Theorem \ref{thm:main} and Proposition~\ref{prop:initial}, respectively). Instead, we 
treat this situation by expanding around the case where all $z_i$ are on the horizontal 
 parts. More precisely, in case that, say, $z_1$ runs on the right vertical part of the contour $\Gamma$ and $z_2, ...., z_4$ are fixed on the horizontal parts, we employ analyticity of $G(z_1)$ away from $[-3,3]$ (since $\|W\|\le 2+\epsilon$ with very high probability). This enables us to write $g(z_1) := \langle G(z_1) A ... G(z_4)B \rangle$ at $z_1  = R + \ii q$ with $q \in [-|t|^{-1}, |t|^{-1}]$, by Taylor expansion around $\zeta := R + \ii |t|^{-1}$, as
\begin{equation} \label{eq:vertexpand}
g(R + \ii q) = \sum_{j=0}^{K}  \frac{(q-|t|^{-1})^j}{2 \pi \ii} \oint_{C_{1/(2|t|)}} \frac{g(w)}{(w - \zeta)^{j+1}} \dd w + \mathcal{O}\big( \widetilde{\mathcal{E}}\big)\,, \quad \widetilde{\mathcal{E}} := \frac{ \langle A^4\rangle^{1/2} \langle B^4\rangle^{1/2} }{R^{K+1} |t|^{K-2}}
\end{equation}
for any $K \in \N$ to be chosen below, where $C_{1/(2|t|)}$ is the circle of radius $1/(2|t|)$ centered around $\zeta$. In \eqref{eq:vertexpand}, we used that the $(K+1)^{\rm th}$ derivative of $g(z_1)$ on the vertical part of the contour is (deterministically) bounded as 
$$\big|g^{(K+1)}(z_1)\big| \lesssim \big|\big\langle \big(G(z_1)\big)^{K+1}A G(z_2) B G(z_3)A G(z_4)B \big\rangle\big| \lesssim \frac{ \langle A^4\rangle^{1/2} \langle B^4\rangle^{1/2} |t|^3}{R^{K+1}},
$$
 since $\Vert G(z_i) \Vert \le |t|$ for $i \in \{2,3,4\}$. Using the representation \eqref{eq:vertexpand}, we can now replace $g(w)$ by $h(w) :=  \langle M(w, A, ... , z_4)B \rangle$, again at the expense of an error $\mathcal{O}_\prec(\mathcal{E})$ by means of Theorem \ref{thm:main}. As one can easily see that $h$ is also analytic around $R+\ii q$ and satisfies the same relation as $g$ in \eqref{eq:vertexpand}, we find that $|g(R+ \ii q) - h(R + \ii q)| \prec \mathcal{E} + \widetilde{\mathcal{E}}$. Hence, in order to absorb $\widetilde{\mathcal{E}}$ into $\mathcal{E}$, it remains to choose $K$ in \eqref{eq:vertexpand} as $K := \lceil \kappa^{-1} \rceil$. 
 
 Therefore, since integrating the $\mathcal{O}_\prec(\mathcal{E})$-bound in \eqref{eq:fourpointint} only adds some $N^\kappa$-factors from the length of $\Gamma$ (note that   $|\ee^{\ii t(z_1 + z_3 - z_2 - z_4)}| \lesssim 1$), which can easily be absorbed into $\prec$, we find that 
 \begin{equation} \label{eq:fourpointintM}
 	\mathcal{F}_{A,B}(t) = \left(\prod_{i \in [4]}\oint_{\Gamma} \frac{\dd z_i}{2 \pi \ii} \right) \ee^{\ii t(z_1 + z_3 - z_2 - z_4)} \langle M(z_1, A, z_2, B, z_3, A,z_4) B \rangle  + \mathcal{O}_\prec \big(\mathcal{E}\big)\,. 
 \end{equation}
 
 Similarly, decomposing $A^2 = (A^2 - \langle A^2 \rangle) + \langle A^2 \rangle$ and analogously for $B$, the two-point part $\mathcal{D}_{A,B}(t)$ from \eqref{eq:OTOCparts} is given by, again for times $1 \le |t| \le N^{(1-\epsilon)/2}$,
 \begin{equation} \label{eq:twopointintM}
\mathcal{D}_{A,B}(t) = \langle A^2 \rangle \langle B^2 \rangle + \left(\prod_{i \in [2]}\oint_{\Gamma} \frac{\dd z_i}{2 \pi \ii}\right)\ee^{\ii t(z_1- z_2)} \langle M(z_1, (A^2 - \langle A^2 \rangle), z_2) (B^2 - \langle B^2 \rangle) \rangle + \mathcal{O}_\prec(\mathcal{E}) \,. 
 \end{equation}

In the other two regimes, $|t|<1$ and $|t|> N^{(1-\epsilon)/2}$, we follow the above arguments for $1 \le |t| \le N^{1/2-\epsilon}$, but replace the contour $\Gamma$ from \eqref{eq:contourOTOC} by\footnote{We choose the first contour only for $t \neq 0$. If $t = 0$, the lhs.~of \eqref{eq:OTOCmain} carries no randomness, and we find $\mathcal{C}_{A,B} (0) = \langle A^2 B^2 \rangle - \langle ABAB \rangle$.} 
\begin{equation*}
\{z \in \C : \mathrm{dist}(z, [-2,2]) = |t|^{-1}\} \quad \text{and} \quad \partial \big( [-R,R] \times \ii [-N^{(1-\epsilon)/2}, N^{(1-\epsilon)/2}] \big)\,, 
\end{equation*}
respectively. This results in error terms in identities analogous to \eqref{eq:fourpointintM}--\eqref{eq:twopointintM}, that can easily be seen to be bounded by $\mathcal{O}_\prec(\mathcal{E})$, just as before, by application of Proposition \ref{prop:initial} and Theorem \ref{thm:main}, respectively. 

It remains to explicitly compute the deterministic terms in \eqref{eq:fourpointintM}--\eqref{eq:twopointintM}, which, by using the formulas \eqref{eq:partrdef}--\eqref{msc dd}, straightforwardly results in the expression given in \eqref{eq:OTOCmain}. 
This concludes the proof of Theorem~\ref{thm:OTOC}. \qed

\section{Zigzag strategy: Proof of Theorem \ref{thm:main}} \label{sec:proof}

In this section we prove our main technical result from Section~\ref{sec:mainSch}, the multi-resolvent local law in Theorem \ref{thm:main}. Its proof is conducted via the \emph{characteristic flow method} \cite{AdhLan, Bourgade2021, HuangLandon, AdhiHuang, LandSos, AggaHuang} followed by a \emph{Green function comparison} (GFT) argument. A combination of these tools, which we coin the \emph{Zigzag strategy},  has first been used in \cite{2210.12060, edgeETH, Gumbel}. It consists of the following three steps: 
\begin{itemize}
\item[\bf 1.] \textbf{Global law.} Proof of a multi-resolvent \emph{global law}, i.e. for spectral parameters ``far away" from the spectrum, $\min_j \mathrm{dist}(z_j, [-2,2]) \ge \delta$ for some small $\delta > 0$ (see Proposition \ref{prop:initial} below). 
\item[\bf 2.] \textbf{Characteristic flow.} Propagate the global law to a \emph{local law} by considering the evolution of the Wigner matrix $W$ along the Ornstein-Uhlenbeck flow 
\begin{equation}
	\label{eq:OUOUOU}
	\dd W_t=-\frac{1}{2}W_t\dd t+\frac{\dd B_t}{\sqrt{N}},  \qquad W_0=W\,,
\end{equation}
with $B_t$ a standard real symmetric or complex Hermitian Brownian motion, thereby introducing an order one Gaussian component (see Proposition \ref{prop:zig}). The spectral parameters of the resolvents evolve from the global regime to the local regime according to the \emph{characteristic (semicircular) flow} 
\begin{equation}
	\label{eq:chardef}
	\partial_t z_{i,t}=-m(z_{i,t})-\frac{z_{i,t}}{2}\,.
\end{equation}
The simultaneous effect of these two evolutions is a key cancellation of two large terms.
\item[\bf 3.] \textbf{Green function comparison.} Remove the Gaussian component by a Green function comparison (GFT) argument (see Proposition \ref{prop:zag}).
\end{itemize}

As the first step, we have the following global law, the proof of which is completely analogous to the proofs presented in \cite[Appendix B]{multiG}, \cite[Appendix A]{A2}, and \cite[Appendix A]{edgeETH}, and so omitted. Proposition \ref{prop:initial} is stated for a general deterministic matrix $B$ since the traceless condition plays no role in this case. 
\begin{proposition}[Step 1: Global law]
\label{prop:initial}
Let $W$ be a Wigner matrix satisfying Assumption~\ref{ass:entries}, and fix $k \in \N$ and $\delta>0$.  Consider spectral parameters $z_1, ... , z_{k} \in \C \setminus \R$, the associated resolvents $G_j = G(z_j) := (W-z_j)^{-1}$, and  deterministic matrices $B_1, ... , B_k \in \C^{N \times N}$. 
Then, uniformly in spectral parameters satisfying $d:=\min_{j \in [k]}\mathrm{dist}(z_j, [-2, 2])\ge \delta$ and deterministic matrices $B_1, ... , B_k$, it holds that
\begin{equation}
\label{eq:mainresultin}
\left| \langle G_1 B_1 ... G_k B_k \rangle - \langle {M}_{[1,k]}B_k \rangle \right| \prec 
\frac{1}{N} \prod_{i \in [k]} \vertiii{B_i}_{2k, d}\,. 
	\end{equation}
\end{proposition}

Next, using Proposition \ref{prop:initial} as an input, we derive Theorem~\ref{thm:main} for Wigner matrices which have an order one Gaussian component, as formulated in Proposition~\ref{prop:zig}. 
For this purpose we consider the evolution of the Wigner matrix $W$ along the Ornstein-Uhlenbeck flow \eqref{eq:OUOUOU} and define its resolvent $G_t(z):=(W_t-z)^{-1}$ with $z_i\in\C\setminus\R$. Even if not stated explicitly we will always consider this flow only for short  times, i.e. for $0\le t\le T$,  where the maximal time $T$ is smaller than some ${\gamma}>0$. Note that the first two moments of $W_t$ are preserved along the flow \eqref{eq:OUOUOU}, and hence the self-consistent density of states of $W_t$ is unchanged;
it remains the standard semicircle law. We now want to compute the deterministic approximation to an alternating product of resolvents and deterministic matrices $A_1, A_2, ...$ with trace zero
\begin{equation}
\label{eq:quantnochar}
{G}_t(z_1)A_1 {G}_t(z_2) A_2 {G}_t(z_3)A_3\dots,
\end{equation}
and have a very precise estimate of the error term.

In fact, we will let the spectral parameters evolve with time with a carefully chosen equation
that will conveniently cancel some error terms in the time evolution of \eqref{eq:quantnochar}. The corresponding equation will be
the characteristic equation for the semicircular flow, i.e. given by the first order ODE  \eqref{eq:chardef} 
(see \cite[Figure~1]{edgeETH} for an illustration of this flow).  \nc
We remark that, along the characteristics we have
\begin{equation}
\label{eq:characteristics}
\partial_t m(z_{i,t})=-\partial_z m(z_{i,t}) \left(m(z_{i,t})+\frac{z_{i,t}}{2}\right)=-\partial_z m(z_{i,t})\left(-\frac{1}{2m(z_{i,t})}+\frac{m(z_{i,t})}{2}\right)=\frac{m(z_{i,t})}{2},
\end{equation}
where in the last two equalities we used the defining equation $m(z)^2 + zm(z)+1=0$ of the Stieltjes transform of the semicircular law. In particular, this implies that $\rho_{i,s}\sim \rho_{i,t}$ for any $0\le s\le t$, where we denoted $\rho_{i,t} := \pi^{-1} |\Im m(z_{i,t})|$. In contrast to that, the behavior of the $\eta_{i,t} := |\Im z_{i,t}|$ depends on the regime: in the bulk $\eta_{i,t}$ decreases linearly in time with a speed of order one, at the edge the decay is still linear but with a speed depending on the size of the local density of states. By standard ODE theory we obtain the following lemma:
 \begin{lemma}[see Lemma~3.2 in \cite{edgeETH}]
 	\label{lem:propchar}
 	Fix an $N$--independent $\gamma>0$, fix $0<T< \gamma \nc$, and pick $z\in \C\setminus \R$.
 	Then there exists an initial condition $z_0$ such that the solution $z_t$ of  \eqref{eq:chardef} with this initial condition $z_0$ satisfies $z_T=z$. Furthermore, there exists a constant $C>0$ such that \nc $\mathrm{dist}(z_0,[-2,2]) \ge C\nc T$.
 \end{lemma}

The spectral parameters evolving by
  \eqref{eq:OUOUOU}  will have the property 
  that 
  \begin{equation}
{G}_t(z_{1,t})A_1 \dots A_{k-1}{G}_t(z_{k,t})-M_{[1,k],t}\approx {G}_0(z_{1,0})A_1\dots A_{k-1}{G}_0(z_{k,0})-M_{[1,k],0},
\end{equation}
with $M_{[1,k],t}:=M(z_{1,t},A_1,\dots, A_{k-1},z_{k,t})$, for any $0\le t \le T$. Note that the deterministic approximation $M_{[1,k],t}$ depends on time only through the time dependence of the spectral parameters, 
 the deterministic approximation of \eqref{eq:quantnochar} with fixed spectral parameters
  does not depend on time, i.e. it is unchanged along the whole flow \eqref{eq:OUOUOU}. 

\begin{proposition}[Step 2: Characteristic flow] \label{prop:zig}
Fix any $\epsilon, \gamma>0$, a time $0\le T\le \gamma$, and $K \in \N$. Consider $z_{1,0},\dots,z_{K,0}\in\C \setminus \R$ as initial conditions to the solution $z_{j,t}$ of \eqref{eq:chardef} for $0 \le t \le T$ and define $G_{j,t}:=G_t(z_{j,t})$ as well as $\eta_{j,t}:=|\Im z_{j,t}|$ and $\rho_{j,t}:=\pi^{-1}|\Im m(z_{j,t})|$. 
 
 Let $k \le K$, define $\ell_t := \min_{j \in [k]} \eta_{j,t} \rho_{j,t}$ and recall \eqref{eq:plSchatten}. Then, assuming that 
\begin{equation}
\label{eq:inass0}
\left| \langle {G}_{1,0} A_1 ... {G}_{k,0} A_{k} \rangle - \langle {M}_{[1,k],0}A_k \rangle \right|\prec \frac{1}{N}\prod_{i \in [k]} \vertiii{A_i}_{2k, \ell_0} \,, 
\end{equation}
holds uniformly for any $k \le K$, any choice of deterministic traceless $A_1, ... , A_k$ and any choice of $z_{i,0}$'s such that $N \ell_0 \ge N^\epsilon$ and $\max_{j\in [k]} |z_{j,0}| \le N^{1/\epsilon}$, then we have
\begin{equation}
\label{eq:flowg}
\left| \langle {G}_{1,T} A_1 ... {G}_{k,T} A_k\rangle - \langle {M}_{[1,k],T}A_k \rangle \right|\prec \frac{1}{N}\prod_{i \in [k]} \vertiii{A_i}_{2k, \ell_T}
\end{equation}
for any $k \le K$, again uniformly in traceless matrices $A_i$ and in spectral parameters satisfying $ N \ell_T  \ge N^{\epsilon}$ and $\max_{j\in [k]} |z_{j,T}| \le N^{1/\epsilon}$.
\end{proposition}
The proof of Proposition~\ref{prop:zig} is given in Section~\ref{sec:optA}.
As the third and final step, we show that the additional Gaussian component introduced in Proposition~\ref{prop:zig} 
can be removed using a Green function comparison (GFT) argument at the price of a negligible error. 
The proof of this proposition is presented in Section~\ref{sec:GFT}.

\begin{proposition}[Step 3: Green function comparison] \label{prop:zag}
Let $H^{(\bm v)}$ and $H^{(\bm w)}$ be two $N \times N$ Wigner matrices with matrix elements given by the random variables $v_{ab}$ and $w_{ab}$, respectively, both satisfying Assumption \ref{ass:entries} and having matching moments up to third order,\footnote{This condition can easily be relaxed to being matching up to an error of size $N^{-2}$ as done, e.g., in \cite[Theorem~16.1]{EYbook}.} i.e.
\begin{equation} \label{eq:momentmatch}
	\E \bar{v}_{ab}^u v_{ab}^{s-u} = \E \bar{w}_{ab}^u w_{ab}^{s-u}\,, \quad s \in \{0,1,2,3\}\,, \quad u \in \{0,...,s\}\,. 
\end{equation}
Fix $K \in \N$ and consider spectral parameters $z_1, ... , z_{K} \in \C \setminus \R$ satisfying $ \min_j N \eta_j \rho_j  \ge N^{\epsilon}$ and $\max_j |z_j| \le N^{1/\epsilon}$ for some $\epsilon > 0$ and the associated resolvents $G_j^{(\#)} = G^{(\#)}(z_j) := (H^{(\#)}-z_j)^{-1}$ with $\# = \bm v, \bm w$. Pick traceless matrices $A_1, ... , A_K \in \C^{N \times N}$. 

Assume that, for $H^{(\bm v)}$, 
we have the following bounds (writing $G_j \equiv G_j^{(\#)}$ for brevity): For any $k \le K$, consider any subset of cardinality $k$ of the $K$ spectral parameters, and similarly, consider any subset of cardinality $k$ of the deterministic matrices. Relabeling both of them by $[k]$, setting $\ell := \min_{j \in [k]} \eta_j \rho_j$ and recalling \eqref{eq:plSchatten}, we have that
	\begin{equation} \label{eq:zagmultiG}
		\left| \langle {G}_1 A_1 ...{G}_k A_k \rangle - \langle {M}_{[1,k]}A_k \rangle \right| \prec \frac{1}{N}\prod_{i \in [k]} \vertiii{A_i}_{2k, \ell}
	\end{equation}
	uniformly in all choices of subsets of $z$'s and $A$'s.

Then, \eqref{eq:zagmultiG} also holds for the ensemble $H^{(\bm{w})}$, uniformly all choices of subsets of $z$'s and $A$'s. 
\end{proposition}

We are now ready to finally conclude the proof of Theorem \ref{thm:main}. Fix $T>0$, and fix $z_1,\dots,z_{k}\in\C \setminus \R$ such that  $ \min_i N \eta_i \rho_i \ge N^\epsilon$, and let $z_{i,0}$ be the initial conditions of the characteristics \eqref{eq:chardef} chosen so that $z_{i,T}=z_i$ (this is possible thanks to Lemma~\ref{lem:propchar}). Then, the assumption \eqref{eq:inass0} of Proposition~\ref{prop:zig} is satisfied  for those $z_{i,0}$ by Proposition~\ref{prop:initial} with $\delta=CT$, since $d \gtrsim \ell_0$ and where $C>0$ is the constant from Lemma~\ref{lem:propchar}. We can thus use Proposition~\ref{prop:zig} to show that \eqref{eq:flowg} holds. Finally, the Gaussian component added in Proposition~\ref{prop:zig} is removed using Proposition~\ref{prop:zag} with the aid of a complex version (see \cite[Lemma A.2]{edgeETH}) of the standard moment-matching lemma \cite[Lemma~16.2]{EYbook}. \qed

\section{Characteristic flow: Proof of Proposition~\ref{prop:zig}} \label{sec:optA}

In this section, we give the proof of Proposition \ref{prop:zig}. The argument is divided into three parts. 
\begin{itemize}
\item[(i)] We begin by introducing several deterministic and stochastic control quantities, which play a fundamental role throughout the rest of this paper. The stochastic control quantities (some normalized differences between a resolvent chain and the corresponding $M$-term, see \eqref{eq:defphi}--\eqref{eq:defpsi} below) satisfy a system of \emph{master inequalities} (see Propositions~\ref{pro:masterinI}--\ref{pro:masterinII}). 
\item[(ii)] Taking these master inequalities together with Proposition \ref{prop:initial} as inputs, we will prove Proposition~\ref{prop:zig} in Section \ref{subsec:closing}, thus concluding Step 2 of the argument in Section \ref{sec:proof}. 
\item[(iii)] Afterwards, based on Proposition \ref{prop:initial} and several relations among the \emph{deterministic} control quantities (see Lemma \ref{lem:m-relnoIM}), we will prove the master inequalities in Section \ref{subsec:pfmasred}.  
\end{itemize}

To keep the presentation simpler, within this section we assume that $ \E \chi^2_{\mathrm{od}}=0$ and that $ \E \chi^2_{\mathrm{d}}=1$, i.e. we consider only complex Wigner matrices. The modifications for the general case are analogous to \cite[Section 4]{edgeETH}, and so omitted.

We recall our choice of the characteristics
\begin{equation}
\partial_t z_{i,t}=-m(z_{i,t})-\frac{z_{i,t}}{2},
\end{equation}
note that $t\mapsto \eta_{i,t}$ is decreasing and that $\rho_{i,t}\sim \rho_{i,s}$ for any $0\le s \le t$ (see \eqref{eq:characteristics} and the paragraph below it for more details). Additionally, we record the following, trivially checkable, simple integration rule  which will be used many times in this section:
\begin{equation}\label{etaint}
\int_0^t \frac{1}{\eta_s^\alpha}\dd s\lesssim \frac{\log N}{\eta_t^{\alpha-2} \ell_t} \quad \text{with} \quad \eta_s := \min_i \eta_{i,s}\,, \quad \ell_s := \min_{i} \eta_{i,s} \rho_{i,s}\,. 
\end{equation}

Along the characteristics, using the short--hand notation $G_{i,t}=(W_t-z_{i,t})^{-1}$, with $W_t$ being the solution of \eqref{eq:OUOUOU}, by It\^{o}'s formula, we have\footnote{We point out that \eqref{eq:flowka} holds for any matrix $A_i\in\C^{N\times N}$, i.e.~we did not use that the $A_i$'s are traceless.}
\begin{equation}
\begin{split}
\label{eq:flowka}
\dd \langle (G_{[1,k],t}-M_{{[1,k]},t})A\rangle&=\frac{1}{\sqrt{N}}\sum_{a,b=1}^N \partial_{ab}  \langle {G}_{[1,k],t}A\rangle\dd B_{ab,t}+\frac{k}{2}\langle {G}_{[1,k],t}A\rangle\dd t +\sum_{i,j=1\atop i< j}^k\langle {G}_{[i,j],t}\rangle\langle {G}_{[j,i],t}\rangle\dd t \\
&\quad+\sum_{i=1}^k \langle G_{i,t}-m_{i,t}\rangle \langle {G}^{(i)}_{[1,k],t}A\rangle\dd t -\partial_t \langle M_{[1,k],t}A\rangle \dd t,
\end{split}
\end{equation}
where $\partial_{ab}$ denotes the directional derivative $\partial_{w_{ab}}$. We also set
\begin{equation}
\label{eq:deflongchaings}
{G}_{[i,j],t}:=\begin{cases}
G_{i,t}A_i\dots A_{j-1}G_{j,t} & \mathrm{if}\quad  i<j \\
G_{i,t} & \mathrm{if}\quad i=j \\
G_{i,t}A_{i,t}\dots G_{k,t}A_kG_{1,t}A_1\dots A_{j-1}G_{j,t} &\mathrm{if}\quad i>j,
\end{cases}
\end{equation}
and defined ${G}^{(l)}_{[i,j],t}$ as $G_{[i,j],t}$ but with the $l$--th resolvent substituted by $G_{l,t}^2$.
The last definition for $i>j$ reflects the cyclicity of the trace, since ${G}_{[i,j],t}$ will be
needed in a tracial situation. For any $i,j\in [k]$, we denote the deterministic approximation of $G_{[i,j],t}$ by $M_{[i,j],t}$, with $M_{[i,j],t}$ being defined as in \eqref{eq:Mdef} with $[1,k]$ replaced by $[i,j]$ if $i<j$ and by $[1,j] \cup [i,k]$ if $i> j$.
\\[2mm]
\textbf{Deterministic control quantities: mean and standard deviation size.} We now introduce two deterministic control quantities that measure the size of the mean of long chains $\langle G_1A_1... G_kA_k \rangle$ and their standard deviation, respectively, in terms of the $\ell$-weighted Schatten norms of $A_1, ... , A_k$ (recall Definition~\ref{def:plSchatten}) and the spectral parameters of the resolvents $G_i = G(z_i)$. 
These are given by 
\begin{equation}
\begin{split}
\label{eq:sizes} 
\m_k(\ell; \bm A_{J_k})
		:= \mathbf{1}(k \ge 2)  
		\prod_{\alpha\in J_k}	  \vertiii{A_\alpha}_{k,\ell} \qquad \text{and} \qquad 
	\s_k(\ell; \bm A_{J_k})
	:= \prod_{\alpha\in J_k}  \vertiii{A_\alpha}_{2k,\ell} \,,
	\end{split}
\end{equation}
and will be called the \emph{mean size} and the \emph{standard deviation size}, respectively. They are functions of a positive number $\ell$ (usually given by $\ell = \min_i \eta_i \rho_i$) and a multiset $\bm A_{J_k} = (A_\alpha)_{\alpha \in J_k}$ of deterministic matrices of cardinality $|J_k| =k$. 

In the following, we shall frequently drop the symbol $\bm A$ from the definitions in \eqref{eq:sizes}, i.e.~write $\m_k(\ell; J_k)$ and $\s_k(\ell; J_k)$ instead of $\m_k(\ell; \bm A_{J_k})$ and $\s_k(\ell; \bm A_{J_k})$, respectively. Moreover, for time dependent spectral parameters, we will also use the notation $\m_k(t)=\m_k(t;J_k) = \m_k(\ell_t; \bm{A}_{J_k})$, with $\ell_t= \min_i \eta_{i,t} \rho_{i,t}$, and a similar notation for the standard deviation size. In particular, we may often omit the dependence on the deterministic matrix $A$. (More generally, we will omit every argument of \eqref{eq:sizes}, whenever it is clear what they are and it does not lead to confusion.) For example, for $i< j$, we have, with $\eta_t := \min_i \eta_{i,t}$, 
\begin{equation}
\label{eq:bMm}
|\langle M_{[i,j],t} A_j\rangle|\lesssim \ell_t \,  \mathfrak{m}_{j-i+1}(t), \qquad\quad |\langle M_{[i,j],t}\rangle|\lesssim (\ell_t/\eta_t) \mathfrak{m}_{j-i}(t).
\end{equation}
The first bound follows from \eqref{eq:Mbound}; the second bound in \eqref{eq:bMm} can easily be obtained by arguments entirely analogous to the ones leading to Lemma \ref{lem:Mboundstrong} in Appendix \ref{sec:addtech}.
The bounds~\eqref{eq:bMm} justify the terminology  that $\mathfrak{m}_k$ is the mean size of a resolvent chain $\langle G_1A_1... G_kA_k \rangle$, which is well approximated by $\langle M_{[1,k]}A_k \rangle$.
The additional $\ell_t$ factor was  introduced to simplify formulas later.

 In the remainder of this section we will often use the following lemma about some
relations among the sizes $\mathfrak{m}_k(t)$, $\s(t)$ from \eqref{eq:sizes}. The proofs are immediate consequences of \eqref{eq:plSchattenrel} and $\ell_s \gtrsim \ell_t$ for $s \le t$ (recall the discussion below \eqref{eq:characteristics}) and hence omitted.
\begin{lemma}[$\mathfrak{m}/\s$-relations]
	\label{lem:m-relnoIM} 
	Let $k \ge 1$ and consider time-dependent spectral parameters $z_{1,t}, ... , z_{k,t}$ and a multiset of traceless matrices $\bm A_{J_k}$ as above.  Then, the mean size $\mathfrak{m}_k$
	and the standard variation size $\s_k$ 
	satisfy the following  inequalities. 
	\begin{itemize}
		\item[(i)] Super-multiplicativity, i.e.~for any $1 \le j \le k-1$ we have
		\begin{equation}
			\label{eq:1in}
				\mathfrak{m}_j (t; J_j)\, \mathfrak{m}_{k-j}(t; J_{k-j})  \lesssim  \, \mathfrak{m}_k (t; J_k)\qquad \text{and} \qquad 
				\s_j (t; J_j)\, \s_{k-j} (t; J_{k-j})\lesssim \, \s_k(t; J_k)
		\end{equation}
		for all disjoint decompositions $J_k = J_j \dot{\cup}J_{k-j}$. 
		\item[(ii)] The mean size can be upper bounded by the standard deviation size as   
		\begin{equation}
			\label{eq:2in}
			\m_k (t; J_k)\lesssim  \,  \s_k (t; J_k) \qquad \text{and} \qquad \sqrt{ \m_{2k}(t; J_k \cup J_k)} \lesssim \s_k(t; J_k)\,,
		\end{equation}
		where $J_k \cup J_k$ denotes the union of the multiset $J_k$ with itself. 
		\item[(iii)] The standard deviation size satisfies the doubling inequality (recall $\ell_t = \min_{j \in [k]} \eta_{j,t} \rho_{j,t}$)
		\begin{equation}
			\label{eq:3in}
			\s_{2k}(t; J_k \cup J_k)
			\lesssim [1+ \sqrt{N \ell_t}] \big(\s_k(t ; J_k)\big)^2\,.
		\end{equation} 
		\item[(iv)] Monotonicity in time: for any $0 \le s \le t$ we have 
		\begin{equation}
			\label{eq:4in}
			\mathfrak{m}_k(s;J_k) \, \ell_s^{\alpha} \lesssim \mathfrak{m}_k(t;J_k) \, \ell_t^{\alpha}\qquad \s_k(s;J_k) \, \ell_s^{\beta} \lesssim \s_k(t;J_k) \, \ell_t^{\beta}
		\end{equation}
		for all $\alpha \in [0,1]$ and $\beta \in [0,1/2]$. 
	\end{itemize}
\end{lemma}
~\\[2mm]
\textbf{Stochastic control quantities: Master and reduction inequalities.}
 Using the notation from \eqref{eq:sizes}, the goal is thus to prove that
\begin{equation}
\label{eq:goalopta}
\langle G_{[1,k],T}A_k \rangle - \langle M_{[1,k],T}A_k \rangle =\langle G_{[1,k],0}A_k \rangle - \langle M_{[1,k],0}A_k \rangle+
\mathcal{O}_\prec \Big(\frac{\s_k(T)}{N}\Big)\,,
\end{equation}
uniformly in the spectrum and uniformly in traceless deterministic matrices $A_i$, for some fixed $T\lesssim 1$. We may henceforth assume that all the $A_i$'s are Hermitian; the general case follows by multilinearity.

For the purpose of proving \eqref{eq:goalopta}, recall the notation $\ell_t := \min_i \eta_{i,t} \rho_{i,t}$ from \eqref{etaint} and define
\begin{equation}
\label{eq:defphi}
\Phi_k(t) = \Phi_k(t; \bm z_{[1,k]}; \bm A_{[1,k]}):= \frac{ N}{ \s_k(\ell_t; \bm A_{[1,k]})} \big|\langle (G_{[1,k],t}-M_{[1,k],t})A_k\rangle \big| \quad \text{for all} \quad k \ge 1\,. 
\end{equation}
 Here, $\bm A_{[1,k]} = (A_1,\dots,A_k)$ and $\bm z_{[1,k]} = (z_1, ... , z_k)$ with $z_i = z_{i,0}$ (initial condition) denote multisets of deterministic matrices and spectral parameters, respectively. We now briefly comment on the definition \eqref{eq:defphi}. We chose the pre--factor in the definition of $\Phi_k(t)$ so that eventually $\Phi_k(t)$ will be of order one with high probability (cf. \eqref{eq:goalopta}).
 However, we will not be able to prove this directly, we first prove that $\Phi_k(t)\prec \sqrt{N\ell_t}$ and then, using this bound as an input, we prove the desired bound $\Phi_k(t)\prec 1$. To implement this technically we introduce another quantity
\begin{equation}
\label{eq:defpsi}
\Psi_k(t) = \Psi_k(t; \bm z_{[1,k]}; \bm A_{[1,k]}):= \frac{ N}{ \s_k(\ell_t; \bm A_{[1,k]})\sqrt{N\ell_t}} \big|\langle (G_{[1,k],t}-M_{[1,k],t})A_k\rangle \big| \quad \text{for all} \quad k \ge 1\,,
\end{equation}
which we will show to be bounded by one (note that $\Psi_k(t)\prec 1$ implies $\Phi_k(t)\prec \sqrt{N\ell_t}$), and then show that $\Psi_k(t)\prec 1$ in fact implies $\Phi_k(t)\prec 1$. 

We start considering $\Psi_k(t)$; note that by \eqref{eq:inass0} it follows  
 \begin{equation}
 \label{eq:initialpsi}
 \Psi_k(0)\lesssim \Phi_k(0)\prec 1,
 \end{equation}
for any $k\ge 1$. For this purpose we will derive a series of {\it master inequalities} for these quantities with the following structure.
We assume that
\begin{equation}\label{psipsi}
\Psi_k(t)\prec \psi_k
\end{equation}
holds
for some deterministic control parameter $\psi_k\ge 1$, uniformly in deterministic matrices $A_i$,  times $0\le t \le T$
and in spectral parameters with $\ell_t\ge N^{-1+\epsilon}$ for some small fixed $\epsilon>0$ (we stress that the $\psi_k$'s depend neither on time, nor on the spectral parameters $z_{i,t}$, nor on the deterministic matrix $A$).
Starting from \eqref{psipsi} we derive an improved upper bound for $\Psi_k(t)$
 and show that, by \emph{iterating} these inequalities, we indeed obtain the desired $\Psi_k(t)\prec 1$. The main inputs to prove this fact are the \emph{master inequalities} in \eqref{eq:masterkaa}, which informally states that if $\Psi_l(t)\prec \psi_l$ for any $l=1,\dots,  k + \mathbf{1}(k \, \mathrm{odd})$, then this actually implies the improved bound \eqref{eq:masterkaa}.

 \begin{proposition}[Master inequalities]
 \label{pro:masterinI}
Fix $k\in \N$. Assume that $\Psi_l(t)\prec \psi_l$ 
for some deterministic control parameters $\psi_l$, \nc 
for any $1\le l\le k+\bm1(k\,\mathrm{odd})$, uniformly in $t\in [0,T]$. Then 
\begin{equation}
\label{eq:masterkaa}
\Psi_k(t)\prec 1+\sum_{j=1}^{k-2}\psi_j+\frac{1}{(N\ell_T)^{1/4}}\sum_{j=1}^k\widetilde{\psi}_j\widetilde{\psi}_{k-j},
\end{equation}
uniformly in $t\in [0,T]$. Here we set $\widetilde{\psi}_l:=\psi_l+\bm1(l\,\mathrm{odd})\sqrt{\psi_{l-1}\psi_{l+1}}$ and $\psi_0:=1$.
\end{proposition}

Using \eqref{eq:masterkaa}, in the next section we show that $\Psi_k(t)\prec 1$
 by an iterative procedure. \nc
 Then, to conclude $\Phi_k(t)\prec 1$, we rely on the following proposition, 
 which will eventually prove Proposition \ref{prop:zig}. \nc
\begin{proposition}
\label{pro:masterinII}
Fix $k\in \N$, assume that $\Phi_l(t)\prec 1$, for $l\le k-2$ uniformly in $t\in [0,T]$, and that $\Phi_l(t)\prec \sqrt{N\ell_t}$ for $1\le l\le 2k$ and $t\in [0,T]$. Then 
\begin{equation}
\label{eq:masterkaaII}
\Phi_k(t)\prec 1,
\end{equation}
uniformly in $t\in [0,T]$.
\end{proposition}

\subsection{Closing the hierarchy: Proof of Proposition \ref{prop:zig}}
\label{subsec:closing} 

To show that Proposition~\ref{pro:masterinI} in fact implies $\Psi_k(t)\prec 1$ we rely on the following procedure, which we refer to as \emph{iteration} (see, e.g., \cite[Lemma~4.1]{iidpaper}).
 
 \begin{lemma}[Iteration]
 \label{lem:iteration}
 Fix $k\in \N$, $T > 0$, and $N$-independent constants $\epsilon,\delta>0$,  $\alpha\in (0,1)$ and $D>0$.  Let $X$ be a random variable
depending on $k$ time dependent spectral parameters $z_{1,t}, ... , z_{k,t}$, $t \in [0,T]$, and recall that $\ell_t = \min_{j\in [k]} \eta_{j, t}\rho_{j, t}$. 
Assume that the a-priori bound $X\prec N^D$ holds uniformly in $t \in [0,T]$, $\ell_t\ge N^{-1+\epsilon}$. 
   Suppose that  there is a deterministic quantity $x$ (may depend on $\ell_T$ and $N$) such that \nc
    for any  fixed $l\in \N$ the fact that $X\prec x$ uniformly for $t \in [0,T]$ and $\ell_t\ge N^{-1+l\epsilon}$ implies\footnote{Here the scalar $A>0$ should not be confused with the matrices $A_i$ which appear throughout the proof.}
\begin{equation}
\label{eq:iterationrep}
X\prec A+\frac{x}{B}+x^{1-\alpha}C^\alpha,
\end{equation}
uniformly for $t \in [0,T]$ and $\ell_t\ge N^{-1+(l+l')\epsilon}$, for some constants\footnote{The constants $A, B, C$  may depend on $N$,
while $l$ and $l'$ are independent of $N$.}
 $l'\in \N$, $B\ge \delta>0$, $A,C>0$.
  Then, iterating \eqref{eq:iterationrep} finitely many times, we obtain
\[
X\prec A+C,
\]
uniformly for $t \in [0,T]$ and $\ell_t\ge N^{-1+(1+Kl')\epsilon}$, for some $K=K(\alpha,D,\delta)$.
\end{lemma}

We are now ready to present the proof of Proposition \ref{prop:zig}.

\begin{proof}[Proof of Proposition \ref{prop:zig}]

The proof of this proposition is divided into two steps: we first prove that $\Psi_k(t)\prec 1$  for any $k$ \nc
 and then use this information as an input to conclude that $\Phi_k(t)\prec 1$  for any $k$. 
 These two bounds are quite different in spirit. The first one is an a-priori bound so 
the iterative procedure behind its proof must be self-improving. This is reflected in the triangular structure of~\eqref{eq:masterkaa}:
the right hand side contains quantities with index at most $k$ (or $k+1$ when $k$ is odd)
and terms with the highest index come with a small prefactor
(recall that $N\ell_T$ is large). This makes the system~\eqref{eq:masterkaa} closable. 
The  iteration is done essentially for each fixed $k$ and then we use an induction on $k$.
Owing to the parity issue in the definition of $\widetilde\psi$, we use a step-two induction, but this is just
a small technicality.
The second bound $\Phi_k(t)\prec 1$ is quite different, since its proof relies on the a-priori bound 
obtained in the first step. The key point is that in order to prove $\Phi_k(t)\prec 1$ for some fixed $k$, we need
to know the a-priori bound $\Phi_l(t)\prec \sqrt{N\ell_t}$ for any $l\le 2k$, i.e. without the a-priori bound
the system of inequalities behind the proof of $\Phi_k(t)\prec 1$ would not be closable.
 This explains why we need to 
proceed in two stages.

We now start we the proof of $\Psi_k(t)\prec 1$. We first prove this for $k=1,2$ and then, using an inductive argument, we show that the same bound holds for any $k\ge 3$. By \eqref{eq:masterkaa} for $k=1,2$ we have
\begin{equation}
\Psi_1(t)\prec 1+\frac{\psi_1+\sqrt{\psi_2}}{(N\ell_T)^{1/4}}, \qquad\quad \Psi_2(t)\prec 1+\frac{\psi_1^2+\psi_2}{(N\ell_T)^{1/4}}
\end{equation}
Using iteration we then obtain (all estimates are uniform in $t\in [0,T]$)
\begin{equation}
\Psi_1(t)\prec 1+\frac{\sqrt{\psi_2}}{(N\ell_T)^{1/4}}, \qquad\quad \Psi_2(t)\prec 1+\frac{\psi_1^2}{(N\ell_T)^{1/4}}.
\end{equation}
Plugging the first inequality into the second one and using iteration, this immediately gives that $\Psi_1(t)+\Psi_2(t)\prec 1$.

Next, we proceed with the induction step. Fix an even $k\in \N$, and assume that $\Psi_l(t)\prec 1$, for $l\le k-2$, then by \eqref{eq:masterkaa} we have
\begin{equation}
\Psi_{k-1}(t)\prec 1+\frac{\psi_{k-1}+\sqrt{\psi_k}}{(N\ell_T)^{1/4}}, \qquad\quad \Psi_2(t)\prec 1+\psi_{k-1}+\frac{\psi_k+\sqrt{\psi_k}}{(N\ell_T)^{1/4}}
\end{equation}
By iteration, we have
\begin{equation}
\Psi_{k-1}(t)\prec 1+\frac{\sqrt{\psi_k}}{(N\ell_T)^{1/4}}, \qquad\quad \Psi_2(t)\prec 1+\psi_{k-1},
\end{equation}
which concludes $\Psi_l(t)\prec 1$, for $l\le k$, by plugging the first inequality into the second one and using iteration once again.

Finally, to conclude $\Phi_k(t)\prec 1$, we proceed by  a step-two  induction on $k$. 
 For $k=1,2$, the assumption $\Phi_l(t)\prec \sqrt{N\ell_t}$, for $l\le 4$, of Proposition~\ref{pro:masterinII} is satisfied, and so we have $\Phi_1(t)+\Phi_2(t)\prec 1$. Then we proceed with the induction step, i.e. for a fixed even $k\in \N$ we assume that $\Phi_l(t)\prec \sqrt{N\ell_t}$, for $l\le 2k$, and $\Phi_l(t)\prec 1$, for $l\le k-2$. Then, by Proposition~\ref{pro:masterinII}, we have $\Phi_l(t)\prec 1$ for $l\le k$; this concludes the  induction step, hence the proof.
\end{proof}

\subsection{Master and reduction inequalities: Proofs of Propositions~\ref{pro:masterinI}--\ref{pro:masterinII}}
\label{subsec:pfmasred}

We first present the proof of Proposition~\ref{pro:masterinI} in detail and then at the end of this section we explain the minor changes to obtain Proposition~\ref{pro:masterinII}.
\begin{proof}[Proof of Proposition~\ref{pro:masterinI}]
To keep the notation simple, from now on we often omit the $t$--dependence from the resolvents $G_i=G_{i,t}$ and from their deterministic approximations $M_{[i,j]}=M_{[i,j],t}$,  
but we will still keep the $t$--dependence in the spectral parameters, in $\Psi_k(t)$
and in the quantities \eqref{eq:sizes}; additionally, we stress that $\psi_k$ does not depend on time.
All estimates are uniform in the time parameters $t, s\in [0,T]$.

Note that in \eqref{eq:defpsi} we defined $\Psi_k(t)$ only for alternating chains of single resolvents and deterministic $A$'s, i.e. no $G_i^k$ appears. However, in \eqref{eq:flowka}   
we naturally get chains involving $G_i^2$. For these terms we use the estimate (recall the estimate for $\langle M_{[i,j]} \rangle$ from \eqref{eq:bMm}):
\begin{equation} 
	\label{eq:shortenchain}
	\big|\langle G_{[i,j]}-M_{[i,j]} \rangle\big|\prec \frac{\psi_{j-i}\mathfrak{s}_{j-i}(t)}{N\eta_t}.
\end{equation}
For $\Im z_{i,t}\Im z_{j,t}>0$ this trivially follows by integral representation
\begin{equation}
\label{eq:intrep}
G_jG_i=\frac{1}{2\pi \ii }\int_\Gamma \frac{G(z)}{(z-z_{i,t})(z-z_{j,t})}\,\dd z,
\end{equation}
for "linearizing" the product of the first and last $G$'s in $\langle G_{[i,j]}\rangle=\langle G_iA_i\dots A_{j-i}G_j\rangle$ after using cyclicity of the trace. Here $\Gamma$ is a contour which lies in the region $\{ z \in \C :|\Im z| \rho(z)\ge \ell_t/2\}$. For  $\Im z_{i,t}\Im z_{j,t}<0$ we "linearize"  by the resolvent identity. Using \eqref{eq:intrep} will change the value of the imaginary part of the spectral parameters, so the domain on which the inequalities below hold (characterized by $N \ell \ge N^{\epsilon}$) may change from time to time, i.e., say $\epsilon \to \epsilon/2$. However, this can happen only finitely many times as it does not affect the $\prec$--bound (see \cite[Figure 2]{iidpaper} for a detailed discussion of this minor technicality).

	To describe the evolution of $M_{[1,k],t}$ in \eqref{eq:flowka} we rely on the following lemma. 
	\begin{lemma}[Lemma 4.8 of \cite{edgeETH}]
		\label{lem:cancM}
		For any deterministic matrices $A_i\in \C^{N\times N}$ (i.e.~not necessarily traceless), we have
		\begin{equation}
			\partial_t \langle M_{[1,k],t}A_k\rangle=\frac{k}{2}\langle M_{[1,k],t}A_k\rangle+\sum_{i,j=1, \atop i<j}^k\langle M_{[i,j],t}\rangle \langle M_{[j,i],t}\rangle. 
		\end{equation}
	\end{lemma}
	Adding and subtracting the deterministic approximation of each term in \eqref{eq:flowka} and using Lemma~\ref{lem:cancM} we thus obtain
		\begin{equation}
		\begin{split}
			\label{eq:flowkmaa}
			\dd \langle (G_{[1,k],t}-M_{{[1,k]},t})A_k\rangle&=\frac{1}{\sqrt{N}}\sum_{a,b=1}^N \partial_{ab}  \langle {G}_{[1,k],t}A_k\rangle\dd B_{ab,t}+\frac{k}{2}\langle {G}_{[1,k],t}A_k\rangle\dd t+\sum_{i,j=1\atop i< j}^k\langle {G}_{[i,j],t}-M_{[i,j],t}\rangle\langle M_{[j,i],t}\rangle \dd t\\
			&\quad+\sum_{i,j=1\atop i< j}^k\langle  M_{[i,j],t}\rangle\langle G_{[j,i],t}-M_{[j,i],t}\rangle\dd t +\sum_{i,j=1\atop i< j}^k\langle {G}_{[i,j],t}-M_{[i,j],t}\rangle\langle G_{[j,i],t}-M_{[j,i],t}\rangle \dd t \\
			&\quad+ \sum_{i=1}^k \langle G_{i,t}-m_{i,t}\rangle \langle {G}^{(i)}_{[1,k],t}A_k\rangle \dd t .
		\end{split}
	\end{equation}

	We point out that, using a simple change of variables, the term $\langle {G}_{[1,k]}A_k\rangle$ amounts to a simple rescaling $e^{kt/2}$,
	so we will ignore it. The quadratic variation of the martingale term in \eqref{eq:flowkmaa} is given by
	\begin{equation}
	\label{eq:quadvarexp}
	\frac{1}{N}\sum_{a,b=1}^N \big|\partial_{ab}\langle G_{[1,k]}A_k\rangle\big|^2\lesssim \sum_{i=1}^k\frac{\langle\Im G_i(A_i\Im G_{i+1}\dots A_{i-1})\Im G_i(A_i\Im G_{i+1}\dots A_{i-1})^*\rangle}{N^2\eta_{i,t}^2},
	\end{equation}
	 where we also used the {\it Ward-identity} $G_iG_i^* = \eta^{-1}_i\Im G_i$. \nc
	Notice that the rhs. of \eqref{eq:quadvarexp} naturally contains chains of $2k$ resolvents. However, to have a closed system of master inequalities for products of resolvents of length $k$, we split the chain of length $2k$ into the product of two 
	chains of length $k$. 	
	For this purpose we use the following \emph{reduction inequality} which will be proven at the end of this section. For any fixed matrices $R,Q\in \C^{N\times N}$, and spectral parameters $z,w\in\C\setminus \R$, we have\footnote{\label{ftn:absval}We point out that the bound $\Psi_k(t)\le\psi_k$ holds also for chains when some resolvents $G$'s are replaced by their absolute value $|G|$. This can be easily seen using the integral representation
\begin{equation} \label{eq:absGintrep}
	|G(E+\ii\eta)|=\frac{2}{\pi}\int_0^\infty \frac{\Im G(E+\ii\sqrt{\eta^2+v^2})}{\sqrt{\eta^2+v^2}}\,\dd v\,,
\end{equation}
	together with the bound \eqref{eq:defpsi} for chains containing only resolvents (see e.g. \cite[Lemma 5.1]{multiG}).}
\begin{equation}
\label{eq:redin}
\langle \Im G(z) Q G(w) R \Im G(z) R^* G (w)^* Q^* \rangle \lesssim N \langle \Im G(z) Q |G(w)| Q^* \rangle \langle \Im G(z) R^* |G(w)| R \rangle.
\end{equation}

We now focus on the case $k$ being even for notational simplicity. Combining \eqref{eq:quadvarexp} with \eqref{eq:redin} used for $z=z_1, w=z_{k/2}$ and
\[
Q = A_1 G_2 \dots G_{k/2-1}A_{k/2-1}\,, \qquad\qquad\quad R  = A_{k/2} G_{k/2+1} \dots G_{k}A_k,
\]
together with $\Psi_k(t)\prec \psi_k$, we obtain the following bound for the quadratic variation
	\begin{equation}
		\label{eq:imGimGaa}
		\frac{1}{N}\sum_{a,b=1}^N \big|\partial_{ab}\langle G_{[1,k],s}A_k\rangle\big|^2\prec \frac{\mathfrak{m}_k(s)^2}{N\eta_s^2}+\frac{N\ell_t\mathfrak{s}(s)^2\psi_k^2}{N^3\eta_s^2}.
	\end{equation}
	Then, using the Burkholder--Davis--Gundy (BDG) inequality, we conclude that the martingale
	 term in \eqref{eq:flowkmaa} is bounded by
	\begin{equation}
		\label{eq:prob2}
		\begin{split}
		\frac{N}{{\mathfrak{s}_k(t)}\sqrt{N\ell_t}}\left[\int_0^t \left(\frac{\ell_s\mathfrak{m}_k(s)^2}{N\eta_s^2}+\frac{N\ell_s\mathfrak{s}_k(s)^2\psi_k^2}{N^3\eta_s^2}\right)\,\dd s\right]^{1/2}&\lesssim \frac{\mathfrak{m}_k(t)}{\mathfrak{s}_k(t)}+\frac{\psi_k}{\sqrt{N}}\left(\int_0^t \frac{1}{\eta_s^2}\,\dd s\right)^{1/2} \\
		&\lesssim1+\frac{\psi_k}{\sqrt{N\ell_t}},
		\end{split}
	\end{equation}
	with very high probability. Here  in the first inequality we used \eqref{eq:4in}, and in the second inequality we used \eqref{eq:2in} for the first term and~\eqref{etaint} for the second term. 
	
	Similarly, for odd $k$, using \eqref{eq:redin} for $z=z_1, w=z_{(k-1)/2}$ and
	\[
	Q = A_1 G_2 \dots G_{(k-3)/2}A_{(k-3)/2}\,, \qquad\qquad\quad R  = A_{(k-1)/2} G_{(k+1)/2} \dots G_{k}A_k
	\]
	we get a bound $1+\sqrt{\psi_{k-1}\psi_{k+1}}/(N\ell_t)^{1/4}$ for the martingale term.  
		
	Next, using \eqref{eq:shortenchain}, we estimate the contribution of the last term in~\eqref{eq:flowkmaa} 
	\begin{equation}
		\label{eq:proba}
		\begin{split}
			\frac{N}{{\mathfrak{s}_k(t)}\sqrt{N\ell_t}}\int_0^t \langle G_{i,s}-m_{i,s}\rangle \langle G_{[1,k],s}^{(i)}A_k\rangle\,\dd s&\prec \frac{N}{{\mathfrak{s}_k(t)}\sqrt{N\ell_t}} \int_0^t\frac{1}{N\eta_s^2} \left(\ell_s\mathfrak{m}_k(s)+\frac{\psi_k\mathfrak{s}_k(s)\sqrt{N\ell_s}}{N}\right)\,\dd s \\
			&\lesssim \frac{1}{\sqrt{N\ell_t}}+\frac{\psi_k}{N\ell_t},
		\end{split}
	\end{equation}
	where in the first inequality we used $|\langle G_{i,s}-m_{i,s}\rangle|\prec (N\eta_s)^{-1}$ from \eqref{eq:singlegllaw}, and in the second inequality we used \eqref{eq:2in}, \eqref{eq:4in} and \eqref{etaint}.

	Using again \eqref{eq:shortenchain}, we estimate the last term in the first line and the first term in the second line of \eqref{eq:flowkmaa} as follows:
	\begin{equation}
		\begin{split}
			\label{eq:prob3aa}
		\frac{N}{{\mathfrak{s}_k(t)}\sqrt{N\ell_t}}&\int_0^t\sum_{i,j=1\atop i< j}^{k-1}\langle M_{[i,j],s}\rangle\langle G_{[j,i],s}-M_{[j,i],s}\rangle\, \dd s \\ &\prec\frac{N}{{\mathfrak{s}_k(t)}\sqrt{N\ell_t}} \int_0^t\sum_{i,j=1\atop i+1< j}^{k-1} \frac{\ell_s\mathfrak{m}_{j-i}(s)}{\eta_s}\frac{\psi_{k-(j-i)}\mathfrak{s}_{k-j+i}(s)\sqrt{N\ell_s}}{N\eta_s}\,\dd s \lesssim \sum_{j=1}^{k-2}\psi_j,
		\end{split}
	\end{equation}
	where in the first inequality we used that $\langle M_{[i,i+1],s}\rangle=0$ and \eqref{eq:bMm} to estimate $\langle M_{[i,j],s}\rangle$, and in the second inequality we used \eqref{eq:1in}--\eqref{eq:2in}, \eqref{eq:4in} and \eqref{etaint}.

	Finally, for the last term in the second line of \eqref{eq:flowkmaa} we estimate:
	\begin{equation}
		\begin{split}
		\label{eq:prob3aaa}
			&\frac{N}{{\mathfrak{s}_k(t)}\sqrt{N\ell_t}} \int_0^t\sum_{i,j=1\atop i< j}^{k-1}\langle G_{[i,j],s}-M_{[i,j],s}\rangle\langle G_{[j,i],s}-M_{[j,i],s}\rangle \,\dd s \\
			&\qquad\qquad\qquad\quad\prec \frac{N}{{\mathfrak{s}_k(t)}\sqrt{N\ell_t}} \sum_{i,j=1\atop i< j}^{k-1}\int_0^t\frac{\psi_{j-i}\mathfrak{s}_{j-i}(s)\sqrt{N\ell_s}}{N\eta_s}\frac{\psi_{k-(j-i)}\mathfrak{s}_{k-j+i}(s)\sqrt{N\ell_s}}{N\eta_s}\,\dd s \\
			&\qquad\qquad\qquad\quad\lesssim  \frac{1}{\sqrt{N\ell_t}}\sum_{j=1}^{k-1}\psi_j\psi_{k-j},
		\end{split}
	\end{equation}
	where in the second inequality we used \eqref{eq:1in}, \eqref{eq:4in} and \eqref{etaint}. Collecting all these bounds, using that $\ell_t\gtrsim \ell_T$ and $N\ell_T\ge 1$, we obtain
\begin{equation*}
\Psi_k(t)\prec 1+\frac{\psi_k+\bm1(k\,\mathrm{odd})\sqrt{\psi_{k-1}\psi_{k+1}}}{(N\ell_T)^{1/4}}+\sum_{j=1}^{k-2}\psi_j +\frac{1}{(N\ell_T)^{1/2}}\sum_{j=1}^{k-1}\psi_j\psi_{k-j}.
\end{equation*}	
Finally, using that $\psi_l\ge 1$ and $N\ell_T\ge 1$, we conclude \eqref{eq:masterkaa}.
\end{proof}
\begin{proof}[Proof of Proposition~\ref{pro:masterinII}]
The proof of \eqref{eq:masterkaaII} is very similar to the one of \eqref{eq:masterkaa}, for this reason we only explain the minor differences. All the terms in \eqref{eq:flowkmaa} are estimated exactly in the same way as in the proof 
of Proposition~\ref{pro:masterinI}  with the exception of the martingale term. 
 In fact, this is the only step where the estimate for $\Psi_k$ (first step) differs
	from the estimate on $\Phi_k$ (second step). Estimating a longer chain by two smaller ones loses a certain 
	$N\ell$ factor; this loss is unavoidable in the first step, but it can be avoided in the second one, once the a-priori bound is
	available. \nc

 More precisely, the estimates in \eqref{eq:proba}, \eqref{eq:prob3aaa}, after multiplying them by a factor $\sqrt{N\ell_t}$, are the same with the only minor difference that we can now use $\psi_l=1$. The estimate \eqref{eq:prob3aa} becomes (recall that $\langle M_{[i,i+1],s}\rangle=0$)
\begin{equation}
\frac{N}{{\mathfrak{s}_k(t)}}\int_0^t\sum_{i,j=1\atop i+1< j}^{k-1}\langle M_{[i,j],s}\rangle\langle G_{[j,i],s}-M_{[j,i],s}\rangle\, \dd s \prec\frac{N}{{\mathfrak{s}_k(t)}} \int_0^t\sum_{i,j=1\atop i+1< j}^{k-1} \frac{\ell_s\mathfrak{m}_{j-i}(s)}{\eta_s}\frac{\Phi_{k-(j-i)}\mathfrak{s}_{k-j+i}(s)}{N\eta_s}\,\dd s \lesssim 1,
\end{equation}
where the only difference with  \eqref{eq:prob3aa} is that now  in the second inequality we used that $\Phi_l(t)\prec 1$, for $l\le k-2$,
 by  the induction \nc assumption. Instead of the bound for the quadratic variation \eqref{eq:quadvarexp}, using that $\Phi_{2k}(t)\prec \sqrt{N\ell_t}$, the estimate in \eqref{eq:prob2} is replaced by
\begin{equation}
\begin{split}
\frac{N}{\mathfrak{s}_k(t)}\left[\int_0^t \frac{\langle\Im G_{1,s}(A_1\dots A_k)\Im G_{1,s}(A_k\dots A_1)^*\rangle}{N^2\eta_s^2}\,\dd s\right]^{1/2}&\prec \frac{N}{\mathfrak{s}_k(t)} \left[\int_0^t \left(\frac{\mathfrak{m}_{2k}(s)\ell_s}{N^2\eta_s^2}+\frac{\Phi_{2k}(s)\mathfrak{s}_{2k}(s)}{N^3\eta_s^2}\right)\,\dd s \right]^{1/2} \\
&\lesssim  \left[\int_0^t \left(\frac{\ell_t}{\eta_s^2}+\frac{\sqrt{N\ell_t}}{N\eta_s^2}\right)\,\dd s \right]^{1/2} \lesssim 1\,,
\end{split}
\end{equation}
where in the second inequality we used \eqref{eq:2in}, \eqref{eq:4in} and in the last inequality we used \eqref{etaint}. This concludes the proof of Proposition~\ref{pro:masterinII}.
\end{proof}

We conclude this section with the proof of \eqref{eq:redin}.

\begin{proof}[Proof of \eqref{eq:redin}]

By spectral decomposition  of $W$ ($\lambda_i, {\bm u}_i$ being its eigenvalues and eigenvectors) \nc we have
\begin{equation}
\begin{split}
\langle \Im G(z) Q G(w) R \Im G(z) R^* G (w)^* Q^* \rangle &=\frac{|\Im z|^2}{N}\sum_{i,j,k,l=1}^N\frac{\langle {\bm u}_i,Q{\bm u}_j\rangle\langle {\bm u}_j,R{\bm u}_k\rangle\langle {\bm u}_k,R^*{\bm u}_l\rangle\langle {\bm u}_l,Q^*{\bm u}_i\rangle}{|\lambda_i-z|^2(\lambda_j-w)|\lambda_k-z|^2(\lambda_l-\overline{w})}  \\
&\lesssim N \langle \Im G(z) Q |G(w)| Q^* \rangle \langle \Im G(z) R^* |G(w)| R \rangle,
\end{split}
\end{equation}
where in the last step we used Schwarz inequality
\[
\big|\langle {\bm u}_i,Q{\bm u}_j\rangle\langle {\bm u}_j,R{\bm u}_k\rangle\langle {\bm u}_k,R^*{\bm u}_l\rangle\langle {\bm u}_l,Q^*{\bm u}_i\rangle\big|\lesssim \big|\langle {\bm u}_i,Q{\bm u}_j\rangle\langle {\bm u}_k,R^*{\bm u}_l\rangle\big|^2+\big|\langle {\bm u}_j,R{\bm u}_k\rangle\langle {\bm u}_l,Q^*{\bm u}_i\rangle\big|^2\,. \qedhere
\]
\end{proof}

\section{Green function comparison: Proof of Proposition \ref{prop:zag}} \label{sec:GFT}
In this section, we remove the Gaussian component introduced in Proposition \ref{prop:zig} using
 a Green function comparison (GFT) argument, i.e.~we prove Proposition \ref{prop:zag}. The basic idea is the same as in \cite[Section~5]{edgeETH}: We perform a \emph{self-consistent} GFT (i.e., given a local law for one ensemble, we aim to prove it for a different one) using an entry-by-entry
  Lindeberg replacement strategy in $O(N^2)$ many steps.  Note that, unlike here,  in typical applications of GFT to 
  answer universality questions, the local law is given as
  an a-priori input. Prior to~\cite{edgeETH} the GFT has been used in a similar spirit by Knowles and Yin \cite{KnowYin}
   in order to prove a single resolvent local law for ensembles, where the deterministic approximation 
   $M$ to $G$ is no longer a multiple of the identity (e.g.~deformed Wigner matrices). In contrast to our approach, they used a continuous interpolation between ensembles, but we stick with the entrywise Lindeberg replacement, which is easier
   to adjust to multiple resolvents, similarly as in \cite{edgeETH}. 

A characteristic property of the Lindeberg strategy, is that along the replacement procedure
  \emph{isotropic resolvent chains} naturally 
arise. In particular, 
we have to consider the isotropic analog 
of Theorem \ref{thm:main}, the average local law, as well, and need to show that also  
\begin{equation} \label{eq:isochain}
\big(G_1A_1...G_k A_k G_{k+1}\big)_{\bm x \bm y}
\end{equation}
concentrates around a deterministic value $\big(M(z_1, A_1, ... , z_k, A_k, z_{k+1})\big)_{\bm x \bm y}$ with $M$ given by \eqref{eq:Mdef}. This will also be done via the Zigzag strategy (recall the outline in the beginning of Section \ref{sec:proof}) in Section \ref{subsec:isoproof}. 

First, analogously to Lemma \ref{lem:Mbound}, we have the following bound on the deterministic approximation of \eqref{eq:isochain}, the proof of which is deferred to Appendix \ref{sec:addtech}. 
\begin{lemma}[Isotropic $M$-bounds] \label{lem:Mboundiso}
	Assume the setting of Lemma \ref{lem:Mbound} but with $k+1$ (instead of $k$) spectral parameters $z_1, ... , z_{k+1} \in \C \setminus \R$. 
	Then, for deterministic vectors $\bm x, \bm y \in \C^N$ with $\Vert \bm x \Vert , \Vert \bm y \Vert \lesssim 1$, it holds that\footnote{Analogously to Footnote \ref{ftn:odd}, we point out that the case of $k$ odd admits an improved bound by $\prod_{i \in [k]}  \Vert A_i \Vert^{\frac{1}{k}}  \,  \vertiii{A}_{\infty, \ell}^{1-\frac{1}{k} }$, but we do not follow this improvement for simplicity. } 
	\begin{equation} \label{eq:Mboundiso}
			\left| \langle \bm{x}, {M}_{[1,k+1]}(z_1, A_1, ... , A_k, z_{k+1} ) \bm{y} \rangle \right| \lesssim  
			 \prod_{i \in [k]} \vertiii{A_i}_{\infty, \ell}\,,
	\end{equation}
	where $\vertiii{A_i}_{\infty, \ell}$ has been introduced in Definition \ref{def:plSchatten}. 
\end{lemma}

The analog of Theorem \ref{thm:main} is the following \emph{isotropic multi-resolvent local law}. The proof is given in Section \ref{subsec:isoproof}. 

\begin{theorem}[Isotropic multi-resolvent local laws] \label{thm:isolaw} Assume the setting of Theorem \ref{thm:main} but with $k+1$ (instead of $k$) spectral parameters $z_1, ... , z_{k+1} \in \C\setminus \R$ and let $\bm{x}, \bm{y} \in \C^N$ be deterministic vectors with $\Vert \bm x\Vert, \Vert \bm y \Vert \lesssim 1$. 
	Then it holds that 
\begin{equation} \label{eq:iso}
		\begin{split}
\left| \langle \bm x, {G}_{1} A_1 ... A_k{G}_{k+1} \bm y \rangle - \langle \bm x, {M}_{[1,k+1]}\bm y \rangle \right|  
\prec  \frac{1}{\sqrt{N\ell}} \prod_{i \in [k]} \vertiii{A_i}_{\infty,\ell}
\,,
		\end{split}
	\end{equation}
		where $\vertiii{A_i}_{\infty, \ell}$ has been introduced in Definition \ref{def:plSchatten}. 
\end{theorem}

Analogously to Theorem \ref{thm:main} and \eqref{eq:SchattenNormHSav}, this unifies and improves the previous 
local laws with operator norm (see \cite[Eq.~(2.11b) in Theorem 2.5]{multiG}) and Hilbert-Schmidt norm (see \cite[Corollary~2.4]{A2}) 
 in the bulk of the spectrum. This follows by estimating (in the relevant $\ell \lesssim 1$ regime)
\begin{equation} \label{eq:SchattenNormHSiso}
	\vertiii{A_i}_{\infty,\ell} \lesssim \begin{cases}
		\dfrac{ \Vert A_i \Vert}{\ell^{1/2}} \qquad &\text{for \  \cite[Eq.~(2.11b) in Theorem 2.5]{multiG}} \\[2mm]
		(N\ell)^{1/2}  \dfrac{\langle |A_i|^2 \rangle^{1/2}}{\ell^{1/2}} \qquad & \text{for \ \cite[Corollary~2.4]{A2}}
	\end{cases}
\end{equation}
in \eqref{eq:iso} for $N \ell > 1$ and every $i \in [k]$ by means of elementary inequalities. 

Note that in the isotropic law \eqref{eq:iso} we use the $\vertiii{A_i}_{\infty, \ell}$ norm instead of the $\vertiii{A_i}_{2k, \ell}$ norm in the corresponding averaged law \eqref{eq:mainresult}. By taking $A_{k+1} := N \bm y \bm x^*$ (assume for simplicity that $\langle \bm x, \bm y \rangle = 0$) in \eqref{eq:mainresult} for $k \to k+1$, it would in fact be possible to obtain an isotropic law immediately from Theorem~\ref{thm:main}. However, the bound provided in \eqref{eq:iso} is stronger than that, as can be seen by means of \eqref{eq:plSchattenrel} and $\vertiii{A_{k+1}}_{2(k+1), \ell} \sim N/(N\ell)^{\frac{1}{2(k+1)}}$, which yield
\begin{equation} \label{eq:Schattencompare}
\frac{1}{\sqrt{N \ell}} \prod_{i \in [k]} \vertiii{A_i}_{\infty, \ell} \lesssim \frac{1}{N} \prod_{i \in [k+1]} \vertiii{A_i}_{2(k+1), \ell} \,. 
\end{equation}
In case that all $A_i$ for $i \in [k]$ have large rank, the lhs.~of \eqref{eq:Schattencompare} is in fact much smaller (by some inverse $(N \ell)$-power) than the rhs.~of \eqref{eq:Schattencompare}.

As for Theorem \ref{thm:main}, we now give a concrete example \eqref{eq:opt1Aiso} how to use Theorem \ref{thm:isolaw} for general (i.e.~not necessarily traceless) matrices. 
\begin{example}
	\label{ex:smallkiso}
	For $k=1$, by \eqref{eq:iso}, we have
	\begin{equation}
		\label{eq:opt1Aiso}
		\left|\langle \bm{x}, G_1B G_2 \bm{y}\rangle - m_1 m_2 B_{\bm x \bm y} - \frac{m_1 m_2 \langle B \rangle \langle \bm x, \bm y \rangle}{1 -m_1 m_2}\right| \prec  \, \frac{|\langle B \rangle|}{\sqrt{N \ell^3}} + \frac{\langle |\mathring{B}|^2\rangle^{1/2} }{\sqrt{N\ell^2}} + \frac{\Vert \mathring{B} \Vert}{\sqrt{N \ell}}\,,
	\end{equation}
for a general matrix $B = \mathring{B} + \langle B \rangle$, completely analogously to \eqref{eq:avexample}. 
\end{example}

\subsection{Isotropic law: Proof of Theorem \ref{thm:isolaw}} \label{subsec:isoproof}
Analogously to the proof of the averaged law, Theorem~\ref{thm:main}, the proof of the isotropic law, 
Theorem \ref{thm:isolaw}, is also conducted via the Zigzag strategy with natural modifications, 
so we will be very brief. The initial step, the global law, has already been proven in \cite[Theorem 2.5]{multiG} (see Eq.~(2.11b) in the $d \ge 1$ regime, which is even stronger than \eqref{eq:mainresultinISO} below). 
\begin{proposition}[Step 1: Isotropic global law]
	\label{prop:initialISO}
	Let $W$ be a Wigner matrix satisfying Assumption~\ref{ass:entries}, and fix $k \in \N$ and $\delta>0$.  Consider spectral parameters $z_1, ... , z_{k+1} \in \C \setminus \R$, the associated resolvents $G_j = G(z_j) := (W-z_j)^{-1}$, and  deterministic matrices $B_1, ... , B_k \in \C^{N \times N}$. 
	Then, uniformly in spectral parameters satisfying $d := \min_{j \in [k+1]}\mathrm{dist}(z_j, [-2, 2])\ge \delta$, deterministic matrices $B_1, ... , B_k$ and deterministic vectors $\bm x, \bm y$ with $\Vert \bm x \Vert , \Vert \bm y \Vert \lesssim 1$, it holds that
	\begin{equation}
		\label{eq:mainresultinISO}
		\left| \big(G_1 B_1 ...  B_k G_{k+1}\big)_{\bm x \bm y} -  \big({M}_{[1,k+1]}\big)_{\bm x \bm y} \right| \prec \frac{\prod_{i \in [k]} \vertiii{B_i}_{\infty, d}}{\sqrt{Nd}} \,. 
	\end{equation}
\end{proposition}

In the second step, the global law 
 is propagated to a local law through the characteristic flow \eqref{eq:OUOUOU}--\eqref{eq:chardef}, thereby introducing an order one Gaussian component. The proof of Proposition \ref{prop:zigISO} is postponed to Appendix \ref{sec:isocase}.  
\begin{proposition}[Step 2: Isotropic characteristic flow] \label{prop:zigISO}
	Fix any $\epsilon, \gamma>0$, a time $0\le T\le \gamma$, and $K \in \N$. Consider $z_{1,0},\dots,z_{K+1,0}\in\C \setminus \R$ as initial conditions to the solution $z_{j,t}$ of \eqref{eq:chardef} for $0 \le t \le T$ and define $G_{j,t}:=G_t(z_{j,t})$ as well as $\eta_{j,t}:=|\Im z_{j,t}|$ and $\rho_{j,t}:=\pi^{-1}|\Im m(z_{j,t})|$. 
	
	Let $k \le K$, define $\ell_t := \min_{j \in [k+1]} \eta_{j,t} \rho_{j,t}$ and recall \eqref{eq:plSchatten}. Assuming that 
		\begin{equation}
		\label{eq:inass0ISO}
		\left|  \big({G}_{1,0} A_1 ...  A_{k} {G}_{k+1,0}\big)_{\bm x \bm y} - \big({M}_{[1,k+1],0}\big)_{\bm x \bm y} \right| \prec \frac{1}{\sqrt{N\ell_0}} \prod_{i \in [k]} \vertiii{A_i}_{\infty,\ell_0}
			\end{equation}
	holds uniformly for any $k \le K$, any choice of deterministic traceless $A_1, ... , A_k$, any choice of $z_{i,0}$'s such that $N \ell_0 \ge N^\epsilon$ and $\max_{j\in [k]} |z_{j,0}| \le N^{1/\epsilon}$, and all deterministic vectors $\Vert \bm x \Vert , \Vert \bm y \Vert \lesssim 1$, then we have
	\begin{equation}
		\label{eq:flowgISO}
		\left|  \big({G}_{1,T} A_1 ...  A_{k} {G}_{k+1,T}\big)_{\bm x \bm y} - \big({M}_{[1,k+1],T}\big)_{\bm x \bm y} \right|\prec \frac{1}{\sqrt{N\ell_T}} \prod_{i \in [k]} \vertiii{A_i}_{\infty,\ell_T}
	\end{equation}
	for any $k \le K$, again uniformly in traceless matrices $A_i$, in deterministic vectors $\bm x, \bm y$ with $\Vert \bm x \Vert, \Vert \bm y \Vert \lesssim 1$ and
	in spectral parameters satisfying $ N \ell_T  \ge N^{\epsilon}$ and $\max_{j\in [k]} |z_{j,T}| \le N^{1/\epsilon}$.
\end{proposition}

In the third and final step, we remove the Gaussian component introduced in Proposition \ref{prop:zigISO} by a GFT argument. The proof of Proposition \ref{prop:zagISO} is given in Section \ref{subsubsec:GFTiso} below.

\begin{proposition}[Step 3: Isotropic Green function comparison] \label{prop:zagISO}
	Let $H^{(\bm v)}$ and $H^{(\bm w)}$ be two $N \times N$ Wigner matrices with matrix elements given by the random variables $v_{ab}$ and $w_{ab}$, respectively, both satisfying Assumption \ref{ass:entries} and having matching moments up to third order, i.e.
	\begin{equation} \label{eq:momentmatchISO}
		\E \bar{v}_{ab}^u v_{ab}^{s-u} = \E \bar{w}_{ab}^u w_{ab}^{s-u}\,, \quad s \in \{0,1,2,3\}\,, \quad u \in \{0,...,s\}\,. 
	\end{equation}
	Fix $K \in \N$ and consider spectral parameters $z_1, ... , z_{K+1} \in \C \setminus \R$ satisfying $ \min_j N \eta_j \rho_j  \ge N^{\epsilon}$ and $\max_j |z_j| \le N^{1/\epsilon}$ for some $\epsilon > 0$ and the associated resolvents $G_j^{(\#)} = G^{(\#)}(z_j) := (H^{(\#)}-z_j)^{-1}$ with $\# = \bm v, \bm w$. Pick traceless matrices $A_1, ... , A_K \in \C^{N \times N}$.

	For any $k \le K$, consider any subset of cardinality $k+1$ of the $K+1$ spectral parameters, and similarly, consider any subset of cardinality $k$ of the deterministic matrices. Relabeling them by $[k+1]$ and $[k]$, respectively, setting $\ell := \min_{j \in [k+1]} \eta_j \rho_j$ and recalling \eqref{eq:plSchatten}, 
	we assume that for $G_j = G_j^{(\bm v)}$ we have
	\begin{equation} \label{eq:zagmultiGISO}
		\left|  \big({G}_{1} A_1 ...  A_{k} {G}_{k+1}\big)_{\bm x \bm y} - \big({M}_{[1,k+1]}\big)_{\bm x \bm y} \right|\prec \frac{1}{\sqrt{N\ell}} \prod_{i \in [k]} \vertiii{A_i}_{\infty,\ell} \,,
	\end{equation}
	uniformly in all choices of subsets of $z$'s and $A$'s and in bounded deterministic vectors $\Vert \bm x \Vert, \Vert \bm y \Vert \lesssim 1$.

	Then, \eqref{eq:zagmultiGISO} also holds for the ensemble $H^{(\bm{w})}$, i.e. for $G_j = G_j^{(\bm w)}$,
	uniformly in all choices of subsets of $z$'s and $A$'s and in bounded deterministic vectors $\Vert \bm x \Vert, \Vert \bm y \Vert \lesssim 1$.
\end{proposition}

Based on Propositions \ref{prop:initialISO}--\ref{prop:zagISO}, the proof of Theorem \ref{thm:isolaw} is concluded in the same way as the proof of Theorem \ref{thm:main} in the end of Section \ref{sec:proof}. \qed

 \subsection{GFT argument: Proof of Propositions \ref{prop:zag} and \ref{prop:zagISO}} 
 The principal idea of the GFT argument is  the same as in \cite[Section~5]{edgeETH}
 and even the detailed argument almost directly translates to our case. The main difference is that we now use the conceptual \emph{mean and standard deviation sizes} $\m/\s$ (recall \eqref{eq:sizes}) and $\m^{\rm iso}/\s^{\rm iso}$ (see \eqref{eq:msisoshort} and \eqref{eq:sizeiso} below) as basic deterministic control quantities; while in \cite[Section~5]{edgeETH} they were not introduced explicitly as they essentially boiled down to simple $N$-powers. More precisely, in view of the bounds \cite[Eqs.~(2.17)--(2.18) and (2.20)--(2.21)]{edgeETH}, we simply replace (ignoring all the $\rho$-factors and using the normalization $\langle |A_i|^2\rangle = 1$ for the just mentioned terms in \cite{edgeETH})
\begin{equation} \label{eq:msreplace}
\frac{N^{k/2} }{N\ell} \longrightarrow \m_k \,, \qquad \frac{N^{k/2} }{(N\ell)^{1/2}} \longrightarrow \s_k \,, \qquad N^{k/2} \longrightarrow \m_k^{\rm iso} \,, \qquad \text{and} \qquad \frac{N^{k/2}}{\ell^{1/2}} \longrightarrow \s_k^{\rm iso}\,. 
\end{equation}
 \nc

 Given this similarity \eqref{eq:msreplace} we will henceforth be very brief and only point out a few minor adaptions of the proof in \cite[Section~5]{edgeETH} to our new conceptual notations in Sections~\ref{subsubsec:GFTiso}--\ref{subsubsec:GFTav}, discussing the isotropic and averaged case, respectively. For further simplicity of notation, we shall drop all irrelevant sub-scripts of resolvents and deterministic matrices, i.e.~write $G = G_j$ and $A = A_j$. Additionally, we will also drop the $\widecheck{G}$- and $G^{(\gamma)}$-notations (see \cite[Eq.~(5.10)]{edgeETH}), indicating the precise step in the replacement procedure, as it will be irrelevant for the modifications discussed below. 
 \subsubsection{Part (a): Proof of the isotropic law, Proposition \ref{prop:zagISO} (cf.~Section 5.2 in \cite{edgeETH})} \label{subsubsec:GFTiso}
We start with some preliminaries. In order to express the bounds in \eqref{eq:Mboundiso} and \eqref{eq:iso} concisely, we employ
 the $\m^{\rm iso}/\s^{\rm iso}$-notation 
 (the \emph{isotropic mean and standard deviation sizes}, analogously to \eqref{eq:sizes} in the average case) to write \eqref{eq:Mboundiso} and \eqref{eq:iso} as 
 	\begin{equation} \label{eq:msisoshort}
 		\left| ( {M}_{[1,k+1]})_{\bm x \bm{y}}  \right| \lesssim   (\vertiii{A}_{\infty, \ell})^k =: \m_k^{\rm iso} \quad \text{and} \quad \left|  \big((GA)^kG\big)_{\bm x \bm y} - \big({M}_{[1,k+1]}\big)_{\bm x \bm y} \right|\prec \frac{(\vertiii{A}_{\infty,\ell})^k}{\sqrt{N\ell}} =: \frac{\s_k^{\rm iso}}{N^{1/2}}  \,,
 \end{equation}
 respectively (the general  definition of $\m^{\rm iso}/\s^{\rm iso}$  is given in~\eqref{eq:sizeiso}). 
  Moreover, we also write $M_{j-i+1} \equiv M_{[i,j]}$ for all $1 \le i < j \le k+1$ with a slight abuse of notation (see \cite[Eq.~(5.33)]{edgeETH}). Lastly, we will heavily use the following relations, proven in parts (i)--(ii) of Lemma \ref{lem:m-relnoIMiso}, 
   		\begin{equation} \label{eq:msisoshort2}
 		\mathfrak{m}^{\rm iso}_j \, \mathfrak{m}_{k-j}^{\rm iso} \lesssim \mathfrak{m}_k^{\rm iso}\,, \qquad 
 		\s^{ \rm iso}_j \, \s^{\rm iso}_{k-j} \lesssim \ell^{-1/2}\s^{\rm iso}_k\,, \quad \text{and} \quad 
 	\mathfrak{m}_k^{\rm iso}\lesssim \ell^{1/2} \,  \s_k^{ \rm iso} \,
 \end{equation}
 for all $k \in \N$ and $0 \le j \le k$ using the conventions $\m_0^{\rm iso} := 1$ and $\s_0^{\rm iso} := \ell^{-1/2}$
 (recall $\ell = \min_i \eta_i \rho_i$).
  With \eqref{eq:msisoshort}--\eqref{eq:msisoshort2} at hand, we can now discuss the two bits of the argument in \cite[Section~5.2]{edgeETH}, which are not completely straightforward to adapt to our current setting (cf.~Case (i) and Case (ii) in 
  \cite[Section~5.2.2]{edgeETH}). 
  We will refer to explicit equation numbers within a longer proof in~\cite{edgeETH} that we do not repeat here, hence 
 the reader needs to be familiar with \cite[Section~5.2]{edgeETH} to follow the details. However, to facilitate a high level 
 understanding without going into details, we recall the novel idea in~\cite{edgeETH}. 
 Traditional GFT proofs via the Lindeberg strategy estimated the change of, say, 
 $\E \big|\big((GA)^kG\big)_{\bm x \bm y}\big|^p$ after {\it each} replacement  ${\bm v}_{ab}\to {\bm w}_{ab}$ and showed that it is 
 bounded by $o(N^{-2})$-times its natural target size, so $N^2$ replacements were affordable. In other words the
 telescopic summation over $(a,b)\in [N]^2$ was done trivially. This does not hold for the self-consistent GFT argument
 in~\cite{edgeETH}: the replacement for some $(a,b)$ pairs  may be too large, but their sum is still affordable.
 We will refer to it as the {\it gain from the summation} idea.
 This is related to our efforts to control the observables in tracial norms; the simplest toy example is the estimate
 \begin{equation}
  \frac{1}{N^2} \sum_{ab} |A_{ab}| \le  N^{-1/2} \langle |A|^2 \rangle^{1/2} 
 \label{ex}
 \end{equation}
 which is optimal. However, if we use the best available bound $|A_{ab}|\le \|A \|$ for each summand individually, then 
 the lhs. of \eqref{ex} is overestimated by $\|A\|$ which is much bigger than the rhs.
 With keeping this idea in mind, we now return to the detailed modifications  in \cite[Section~5.2.2]{edgeETH}.

The terms considered in Case (i) are the ones  for which there are enough $(G-M)$-type terms to balance the $N$-prefactor (cf.~\eqref{eq:d=1exampleCase1Schatten} and \eqref{eq:d=1exampleCase2Schatten})  by
solely using \eqref{eq:msisoshort2}. The terms in Case (ii) require to \emph{gain from the summation} over all steps in the replacement procedure. 
\begin{itemize}
\item[Case (i):] First, the analog of \cite[Eq.~(5.34)]{edgeETH} becomes (neglecting the irrelevant $ |\Psi_k|^{p-1}$- and $N^\xi$-factors)
	\begin{equation} \label{eq:d=1exampleCase1Schatten}
	\begin{split}
		&  \frac{N^{1/2}}{\s_k^{\rm iso}}  \sum_{\substack{0\le k_l \le k-1 : \\
				\sum_{l} k_l = k}}\E \Big[\big| \big( ({G} A)^{k_1} {G} - M_{k_1+1}\big)_{\bm x \bm{e}_i}  \big(M_{k_2+1}\big)_{\bm e_j \bm{e}_j} \big( M_{k_3+1}\big)_{\bm e_i \bm{e}_i} \big( M_{k_4+1}\big)_{\bm e_j \bm{e}_j} \big( M_{k_5+1}\big)_{\bm e_i \bm{y}}  \big|  \\
		& \quad  + \big| \big( ({G} A)^{k_1} {G} - M_{k_1+1}\big)_{\bm x \bm{e}_i}  \big(({G} A)^{k_2} {G} -M_{k_2+1}\big)_{\bm e_j \bm{e}_j} \big( M_{k_3+1}\big)_{\bm e_i \bm{e}_i} \big( M_{k_4+1}\big)_{\bm e_j \bm{e}_j} \big( M_{k_5+1}\big)_{\bm e_i \bm{y}}  \big| + 	... \Big] \\
		\lesssim \, &  \frac{N^{1/2}}{\s_k^{\rm iso}} \sum_{\substack{0\le k_l \le k-1 : \\
				\sum_{l} k_l = k}} \left[ \frac{\s_{k_1}^{\rm iso}}{N^{1/2}} \left(\prod_{l =2}^5 \m_{k_l}^{\rm iso}\right)+ \frac{\s_{k_1}^{\rm iso}\s_{k_2}^{\rm iso}}{N} \left(\prod_{l =3}^5 \m_{k_l}^{\rm iso}\right) + ... \right]\lesssim  \left[1 + \frac{1}{(N \ell)^{1/2}}+ ... \right] \lesssim 1\,. 
	\end{split}
\end{equation}
In the penultimate step we used the relations in \eqref{eq:msisoshort2} multiple times to estimate $\s_{k_1}^{\rm iso}\prod_{l =2}^5 \m_{k_l}^{\rm iso} \lesssim \s_k^{\rm iso}$ and $\s_{k_1}^{\rm iso} \s_{k_l}^{\rm iso}\prod_{l =3}^5 \m_{k_l}^{\rm iso} \lesssim \ell^{-1/2}\s_k^{\rm iso}$, and similarly for the other analogous terms indicated by dots. The last step in \eqref{eq:d=1exampleCase1Schatten} is due to $N \ell > 1$. 
\item[Case (ii):] Second, the key trick in \cite[Section~5.2.2]{edgeETH}, the \emph{gain from summations} described in \cite[Example~5.7]{edgeETH}, turns into the following. The \emph{trivial estimate}, cf.~\cite[Eq.~(5.36)]{edgeETH}, reads (again neglecting the irrelevant $ |\Psi_k|^{p-1}$- and $N^\xi$-factors)
	\begin{equation} \label{eq:d=1exampleCase2Schatten}
	\begin{split}
		&\frac{N^{1/2}}{\s_k^{\rm iso}}  \sum_{\substack{0\le k_l \le k: \\
				\sum_{l} k_l = k}}\left[\big| \big( M_{k_1+1}\big)_{\bm x \bm{e}_i}  \big(M_{k_2+1}\big)_{\bm e_j \bm{e}_j} \big( M_{k_3+1}\big)_{\bm e_i \bm{e}_i} \big( M_{k_4+1}\big)_{\bm e_j \bm{e}_j} \big( M_{k_5+1}\big)_{\bm e_i \bm{y}}  \big| + ... \right] \\
		\lesssim \, & \frac{N^{1/2}}{\s_k^{\rm iso}}  \sum_{\substack{0\le k_l \le k: \\
				\sum_{l} k_l = k}} \left[ \left(\prod_{l =1}^5 \m_{k_l}^{\rm iso}\right) + ... \right]\lesssim  \left(N \ell \right)^{1/2}\,, 
	\end{split}
\end{equation}
where we again used \eqref{eq:msisoshort2} multiple times. Analogously to \cite[Eq.~(5.37)]{edgeETH}, this can be improved on by averaging over all replacement steps: Fixing one constellation of $k_l$'s in \eqref{eq:d=1exampleCase2Schatten}, we find
 	\begin{equation} \label{eq:d=1exampleCase2SUMSchatten}
	\begin{split}
		&\frac{N^{1/2}}{\s_k^{\rm iso}} \frac{1}{N^2} \sum_{i,j } \left[\big| \big( M_{k_1+1}\big)_{\bm x \bm{e}_i}  \big(M_{k_2+1}\big)_{\bm e_j \bm{e}_j} \big( M_{k_3+1}\big)_{\bm e_i \bm{e}_i} \big( M_{k_4+1}\big)_{\bm e_j \bm{e}_j} \big( M_{k_5+1}\big)_{\bm e_i \bm{y}}  \big| + ... \right] \\
		& \qquad \qquad \lesssim \, (N\ell )^{1/2}  \frac{1}{N^2} \sum_{i,j } \left[ \frac{\big| (M_{k_1+1})_{\bm x \bm e_i} \big|}{\m_{k_1}^{\rm iso}}  + ... \right] \lesssim (N \ell)^{1/2} \frac{1}{N^{1/2}} \lesssim 1
	\end{split}
\end{equation}
in the relevant $\ell \le 1$ regime (the opposite regime being covered by Proposition \ref{prop:initialISO}). To go to the second line, we employed \eqref{eq:msisoshort2}. In the penultimate step, we used 
 	\begin{equation} \label{eq:schwarzfirstSchatten}
	\sum_i \big| (M_{k_1+1})_{\bm x \bm e_i} \big|  \le \sqrt{N}\sqrt{ \big(| M_{k_1+1}|^2 \big)_{\bm x \bm x} } \lesssim \sqrt{N} \, \m_{k_1}^{\rm iso}\,, 
\end{equation}
which follows by a Schwarz inequality, completely analogously to \cite[Eq.~(5.38)--(5.39)]{edgeETH}. 
\end{itemize}
With these two slight adjustments, which 
 rest on the estimates in \eqref{eq:msisoshort2}, the proof in \cite[Section~5.2]{edgeETH} can entirely be translated to our current setting in a straightforward way, so we omit further details. This concludes the proof of Proposition~\ref{prop:zagISO}. \qed 
\subsubsection{Part (b): Proof of the averaged law, Proposition \ref{prop:zag} (cf.~Section 5.3 in \cite{edgeETH})} \label{subsubsec:GFTav}
For the averaged law, in addition to the bounds in \eqref{eq:msisoshort}--\eqref{eq:msisoshort2}, we will also use the relations
		\begin{equation}
	\label{eq:msavisoshort}
		\m_k  \lesssim 	\m_k^{\rm iso}\,, \qquad \s_k  \lesssim \ell^{1/2} \, 	\s_k^{ \rm iso} \,, \quad \text{and} \quad 
	\s_k^{ \rm iso}  \lesssim \, N \,  \big[1+ (N\ell)^{-1/2}\big]\,  \s_k 
\end{equation}
from \eqref{eq:1inaviso}--\eqref{eq:2inaviso} in Lemma \ref{lem:m-relnoIMaviso} for all $k \in \N$. Just as in the case of Section \ref{subsubsec:GFTiso}, there are two bits of the argument in \cite[Section~5.3]{edgeETH}, that are not entirely straightforward to adjust to the setting of the current paper (cf.~Case (i) and Case (ii) in \cite[Section~5.3]{edgeETH}). 

Analogously to the isotropic case in Section \ref{subsubsec:GFTiso}, while the terms considered in Case (i) solely use \eqref{eq:msisoshort2} and \eqref{eq:msavisoshort}, the terms in Case (ii) require to \emph{gain from the summation} over all steps in the replacement procedure. 

\begin{itemize}
\item[Case (i):] First, consider \cite[Eq.~(5.54)]{edgeETH} with $d=2$ (for concreteness). Then, the estimate turns into (neglecting the irrelevant $N^\xi$-factor)
\begin{equation*}
\left(\frac{N}{\s_k}\right)^2 \frac{1}{N^2}\sum_{\substack{k_1 + k_2 = k \\
		k_1' + k_2' = k}} \left[ \m_{k_1}^{\rm iso} \frac{\s_{k_2}^{\rm iso}}{N^{1/2}} \m_{k_1'}^{\rm iso} \frac{\s_{k_2'}^{\rm iso}}{N^{1/2}} + \m_{k_1}^{\rm iso} \m_{k_2}^{\rm iso} \frac{\s_{k_1'}^{\rm iso}}{N^{1/2}}\frac{\s_{k_2'}^{\rm iso}}{N^{1/2}}+ ...  \right] \lesssim  \left[1 + \frac{1}{N \ell}\right] \lesssim 1 \,, 
\end{equation*}
where we used \eqref{eq:msisoshort2} and \eqref{eq:msavisoshort} to bound $\m_{k_1}^{\rm iso} \s_{k_2}^{\rm iso} \lesssim \ell^{-1/2}[1+(N\ell )^{1/2} ] \s_k$ as well as $\m_{k_1}^{\rm iso} \m_{k_2}^{\rm iso} \lesssim [1+ (N \ell)^{1/2}] \s_k $ and $\s_{k_1'}^{\rm iso} \s_{k_2'}^{\rm iso} \lesssim \ell^{-1}[1+ (N \ell)^{1/2}] \s_k $ in the penultimate step, and the last step used $N \ell > 1$. 
\item[Case (ii):] Second, we again discuss how to \emph{gain from the summation}, as originally explained in \cite[Example~5.11]{edgeETH}. The \emph{trivial estimate} from \cite[Eq.~(5.55)]{edgeETH}, again neglecting the irrelevant $|\Psi_k|^{p-1}$ factor and fixing one constellation of $k_l$'s summing up to $k$, becomes (neglecting the irrelevant $N^\xi$-factor)
 	\begin{equation} \label{eq:d=1exampleCase2AVSchatten}
	\begin{split}
		&(\s_k)^{-1} \left[\big| \big( M_{k_1+1}\big)_{\bm e_i \bm{e}_i}  \big(M_{k_2+1}\big)_{\bm e_j \bm{e}_j} \big( M_{k_3+1}\big)_{\bm e_i \bm{e}_i} \big( M_{k_4+1}\big)_{\bm e_j \bm{e}_j}  \big| + ... \right] \\
		& \qquad \qquad \qquad \lesssim \, (\s_k)^{-1} \left[ \prod_{l}\m_{k_l}^{\rm iso} + ... \right] \lesssim  (N\ell )^{1/2} \,, 
	\end{split}
\end{equation}
where we used \eqref{eq:msisoshort2} and \eqref{eq:msavisoshort}. Analogously to \cite[Eq.~(5.59)]{edgeETH}, we can again improve upon this by averaging over all replacement positions. The key for this \emph{gain from the summation} is the bound
\begin{equation} \label{eq:sumgainAV}
\langle |M_{k+1}|^2 \rangle \lesssim (\s_k)^2 \qquad \forall k \in \N\,, 
\end{equation}
which can be obtained completely analogously to \cite[Lemma~5.12]{edgeETH}. Armed with \eqref{eq:sumgainAV}, the analog of the improved bound \cite[Eq.~(5.59)]{edgeETH} now reads (again neglecting the irrelevant $N^\xi$-factor)
 \begin{equation} \label{eq:d=1exampleCase2AVSUMSchatten}
	\begin{split}
		&(\s_k)^{-1}\frac{1}{N^2}\sum_{i,j}\left[\big| \big( M_{k_1+1}\big)_{\bm e_i \bm{e}_i}  \big(M_{k_2+1}\big)_{\bm e_j \bm{e}_j} \big( M_{k_3+1}\big)_{\bm e_i \bm{e}_i} \big( M_{k_4+1}\big)_{\bm e_j \bm{e}_j}  \big| + ... \right] \\
		\lesssim &(\s_k)^{-1} \left[\prod_{l \in [4]} \left(\frac{1}{N}\sum_{i} \big|\big( M_{k_l+1}\big)_{\bm e_i \bm{e}_i}\big|^2 \right)^{1/2} + ... \right] 
		\lesssim (\s_k)^{-1}\left[\prod_{l \in [4]}  \s_{k_l}+ ... \right] 
		\lesssim 1\,, 
	\end{split}
\end{equation}
where in the first step we employed a trivial Schwarz inequality. 
\end{itemize}

With the above two slight adjustments at hand, which basically rest on the estimates in \eqref{eq:msisoshort2}, \eqref{eq:msavisoshort} and \eqref{eq:sumgainAV}, we can straightforwardly translate the entire proof of the averaged part in \cite[Section~5.3]{edgeETH} to our current setting. This finishes our discussion of the adjustments of the arguments from \cite[Section~5.2]{edgeETH} and thus the proof of Proposition~\ref{prop:zag}. \qed

\appendix

\section{Additional technical results}
\label{sec:addtech}
In this section we prove several additional technical results which were used in the main sections.
\subsection{Bound on the deterministic approximation: Proofs of Lemmas \ref{lem:Mbound} and \ref{lem:Mboundiso}}
We will deduce Lemmas \ref{lem:Mbound} and \ref{lem:Mboundiso} from the following stronger bound. Part (a) of Lemma \ref{lem:Mboundstrong} is already proven in \cite[Lemma A.1]{edgeETH} applied to the special case $\mathfrak{I}_k = \emptyset$. 
The proof of part (b) is analogous to that and hence omitted. 
\begin{lemma}[cf.~Lemma A.1 in \cite{edgeETH}]\label{lem:Mboundstrong}
Fix $k \ge 1$ and let $A_1, ... , A_k \in \C^{N \times N}$ be traceless deterministic matrices.
\begin{itemize}
\item[(a)]  Consider spectral parameters $z_1, ... , z_{k} \in \C \setminus \R$ and define $\eta := \min_{j \in [k]} |\Im z_j|$. 
Then, for every $1 \le s \le \lfloor k/2 \rfloor$ and $\pi \in \mathrm{NC}([k])$ with $|\pi| = k+1-s$, it holds that
\begin{equation} \label{eq:Mboundstrong}
	\left| 
		\langle \mathrm{pTr}_{K(\pi)}(A_1,\ldots,A_{k-1})A_k \rangle \prod_{S\in\pi} m_\circ[S] \right| \lesssim  \frac{1}{\eta^{s-1}} \left(\prod_{\substack{S \in K(\pi) \\ |S| \ge 2}} \prod_{j \in S} \langle |A_j|^{|S|} \rangle^{\frac{1}{|S|}}\right)  \,. 
\end{equation}
with $m_\circ[S]$ being defined in \eqref{eq:freecumulant}. For $s > \lfloor k/2 \rfloor$ the lhs.~of \eqref{eq:Mboundstrong} equals zero. 
\item[(b)] Consider spectral parameters $z_1, ... , z_{k+1} \in \C \setminus \R$ and define $\eta := \min_{j \in [k+1]} |\Im z_j|$. Then, for every $1 \le s \le \lceil (k+1)/2 \rceil$ and $\pi\in\mathrm{NC}([k+1])$ with $|\pi| = k+1-s$, it holds that
\begin{equation} \label{eq:MboundstrongISO}
	\left\Vert  
	\mathrm{pTr}_{K(\pi)}(A_1,\ldots,A_{k})  \prod_{S\in\pi} m_\circ[S] \right\Vert \lesssim  \frac{1}{\eta^{s-1}} \left(\prod_{\substack{S \in K(\pi) \setminus \mathfrak{B}: \\ |S| \ge 2}} \, \prod_{j \in S} \langle |A_j|^{|S|} \rangle^{\frac{1}{|S|}} \right) \left(\prod_{j \in \mathfrak{B}\setminus \{k+1\}} \Vert A_j \Vert \right)\,. 
\end{equation}
where $\mathfrak{B} \equiv \mathfrak{B}(k+1) \in K(\pi)$ being the unique block containing $k+1$ (recall \eqref{eq:partrdef}). For $s > \lceil (k+1)/2 \rceil$ the lhs.~of \eqref{eq:MboundstrongISO} equals zero. 
\end{itemize}
\end{lemma}
Given Lemma \ref{lem:Mboundstrong}, analogously to \cite[Appendix A.1]{edgeETH}, the proofs of Lemmas \ref{lem:Mbound} and \ref{lem:Mboundiso} immediately follow after realizing that $\eta \gtrsim \ell$ and applying Hölder's and Young's inequality. 
To show this mechanism,  consider \eqref{eq:Mboundstrong}
for an example where  $k=6$ and $s=2$, and focus on the non-crossing partition $\pi = \{ 15|2|3|4|6\}$ for which $K(\pi) = \{1234|56\}$. In this case, the rhs.~of \eqref{eq:Mboundstrong} can be estimated as
\begin{equation*}
	\begin{split}
\eta^{-1} \left(\prod_{i=1}^4 \langle |A_i|^4 \rangle^{1/4}\right) \left(\prod_{i=5}^6 \langle |A_i|^2 \rangle^{1/2}\right) &\lesssim 
\ell \left(\prod_{i=1}^4 \big[\langle |A_i|^6 \rangle^{1/8} \ell^{-1/8}\big] \,  \big[\langle |A_i|^2 \rangle^{1/8}\ell^{-1/8}\big]\right) \left(\prod_{i=5}^6 \langle |A_i|^2 \rangle^{1/2}\ell^{-1/2}\right) \\
&\lesssim \ell \prod_{i\in [6]} \left(\frac{\langle |A_i|^6 \rangle^{1/6}}{\ell^{1/6}} +  \frac{\langle |A_i|^2 \rangle^{1/2}}{\ell^{1/2}}\right) \,. 
	\end{split}
\end{equation*}
In the first step, we employed $\eta \gtrsim \ell$ and Hölder's inequality in the form $\langle |A_i|^4\rangle^{1/4} \le \langle |A_i|^6 \rangle^{1/8} \langle |A_i|^2 \rangle^{1/8}$ for $i \in [4]$. In the second step, we then used Young's inequality for $i \in [4]$ and added $\langle |A_i|^6 \rangle^{1/6}\ell^{-1/6}$ for $i =5,6$ to complete the $\vertiii{A_i}_{6,\ell}$ norm for all $i \in [6]$ (recall \eqref{eq:plSchatten} and \eqref{eq:Mbound}). \qed

\subsection{GFT and isotropic local law: Additional proofs for Section \ref{sec:GFT}}

\subsubsection{Isotropic law: Proof of Theorem \ref{thm:isolaw}}
\label{sec:isocase}
In this section we want to study chains of the form
\begin{equation}
	\label{eq:objectiso}
	(G_1A_1\dots A_kG_{k+1})_{{\bm x}{\bm y}}:=\langle {\bm x}, G_1A_1\dots A_kG_{k+1} {\bm y}\rangle,
\end{equation}
for unit deterministic vectors ${\bm x}, {\bm y}$.

Following the notation \eqref{eq:deflongchaings}, by $G_{[1,k+1],t}$ we denote the evolution of the quantity in \eqref{eq:objectiso} along the Ornstein-Uhlenbeck flow \eqref{eq:OUOUOU} with the characteristic equation \eqref{eq:chardef}. Then, by \eqref{eq:flowka} and Lemma~\ref{lem:cancM} (used for $k\to k+1$), choosing $A_{k+1}=N{\bm y}{\bm x}^*$, we obtain the flow
\begin{equation}
	\begin{split}
		\label{eq:flowkaiso}
		\dd (G_{[1,k+1],t}-M_{{[1,k+1]},t})_{{\bm x}{\bm y}}&=\frac{1}{\sqrt{N}}\sum_{a,b=1}^N \partial_{ab}  ({G}_{[1,k+1],t})_{{\bm x}{\bm y}}\dd B_{ab,t}+\frac{k}{2} (G_{[1,k+1],t}-M_{{[1,k+1]},t})_{{\bm x}{\bm y}}\dd t \\
		&\quad+\sum_{i,j=1\atop i< j}^{k+1}\langle {G}_{[i,j],t}-M_{[i,j],t}\rangle ({M}_{[1,i]\cup [j,k+1],t})_{{\bm x}{\bm y}}\dd t \\
		&\quad+\sum_{i,j=1\atop i< j}^{k+1}\langle {M}_{[i,j],t}\rangle ({G}_{[1,i]\cup [j,k+1],t}-M_{[1,i]\cup [j,k+1],t})_{{\bm x}{\bm y}}\dd t  \\
		&\quad+\sum_{i,j=1\atop i< j}^{k+1}\langle G_{[i,j],t}-{M}_{[i,j],t}\rangle ({G}_{[1,i]\cup [j,k+1],t}-M_{[1,i]\cup [j,k+1],t})_{{\bm x}{\bm y}}\dd t  \\
		&\quad+\sum_{i=1}^{k+1} \langle G_{i,t}-m_{i,t}\rangle  ({G}^{(i)}_{[1,k+1],t})_{{\bm x}{\bm y}}\dd t,
	\end{split}
\end{equation}
where we used the notation
\[
({G}_{[1,i]\cup [j,k+1],t})_{{\bm x}{\bm y}}:= (G_{1,t}A_1\dots A_{i-1}G_{i,t}G_{j,t}A_j\dots A_{k-1}G_{k+1,t})_{{\bm x}{\bm y}}
\]
and similarly for $M_{[1,i]\cup [j,k+1],t}$. 
\\[2mm]
\noindent\textbf{Deterministic control quantities.} Next, we introduce the new \emph{isotropic} control quantities,
 the analogues of the \emph{averaged} quantities defined in~\eqref{eq:sizes} (recall \eqref{eq:plSchatten} for the definition of $\vertiii{A_\alpha}_{\infty, \ell}$): 
\begin{equation} \label{eq:sizeiso}
	\begin{split}
		\mathfrak{m}_k^{\rm iso}(\ell; \bm{A}_{J_k}) :=  \prod_{\alpha \in J_k} \vertiii{A_\alpha}_{\infty, \ell} 
	\qquad \text{and} \qquad 
		\s_k^{\rm iso}(\ell; \bm{A}_{J_k}) := \ell^{-1/2}\prod_{\alpha \in J_k} \vertiii{A_\alpha}_{\infty, \ell} \,,
	\end{split}
\end{equation}
which will be called the \emph{isotropic mean size} and \emph{isotropic standard deviation size}, respectively. They are functions of a positive number $\ell$ (usually given by $\ell = \min_i \eta_i \rho_i$ for some spectral parameters $z_i$) and a multiset $\bm A_{J_k} = (A_\alpha)_{\alpha \in J_k}$ of deterministic matrices of cardinality $|J_k| =k$. 

 Similarly to the paragraph below \eqref{eq:sizes}, we may often omit the arguments of  $\mathfrak{m}_k^{\rm iso}$ and $\s_k^{\rm iso}$, write $J_k$ instead of $\bm A_{J_k}$, and for time dependent spectral parameters we use the short--hand notations $\mathfrak{m}_k^{\rm iso}(t) = \m_k^{\rm iso}(t;J_k) = \m_k^{\rm iso}(\ell_t; \bm A_{J_k})$ and $\s_k^{\rm iso}(t) = \s_k^{\rm iso}(t;J_k) = \s_k^{\rm iso}(\ell_t; \bm A_{J_k})$. Note that with this definitions, by \eqref{eq:Mboundiso}, we have (analogously to \eqref{eq:bMm}), with $\eta_t := \min_i \eta_{i,t}$,
\begin{equation}
	\label{eq:mboundiso}
	\big| (M_{[i,j],t})_{{\bm x}{\bm y}}\big|\lesssim \mathfrak{m}_{j-i}^{\rm iso}(t) \qquad \text{and} \qquad \big| (M_{[1,i] \cup [j,k+1], t})_{\bm x \bm y} \big| \lesssim \eta_t^{-1} \m_{k-(j-i)}^{\rm iso}(t)\,.
\end{equation}

We now record the following relations about $\mathfrak{m}_k^{\rm iso}/\s_k^{\rm iso}$, analogously to Lemma \ref{lem:m-relnoIM}.
The proof is again an immediate consequence of \eqref{eq:plSchattenrel} and $\ell_s \gtrsim \ell_t$ for $s \le t$ (recall the discussion below \eqref{eq:characteristics}) and hence omitted. 

\begin{lemma}[$\mathfrak{m}^{\rm iso}/\s^{\rm iso}$-relations] \label{lem:m-relnoIMiso}
		Let $k \ge 1$ and consider time-dependent spectral parameters $z_{1,t}, ... , z_{k+1,t}$ and a multiset 
		 of traceless matrices $\bm A_{J_k}$ as above. Set $\ell_t := \min_{j \in [k+1]} \eta_{j,t} \rho_{j,t}$. 
		 Then, the mean size $\mathfrak{m}^{\rm iso}_k $ and the standard deviation size $\mathfrak{s}_k^{ \rm iso}$ 
	satisfy the following  inequalities (using the conventions $\m_0^{\rm iso}(t) := 1$ and $\s_0^{\rm iso}(t) := \ell_t^{-1/2}$). 
	\begin{itemize}
		\item[(i)] Super-multiplicativity; i.e. for any $0 \le j \le k$ it holds that
		\begin{equation}
			\label{eq:1iniso}
			\begin{split}
				\mathfrak{m}^{\rm iso}_j ({t;J_{j}})\, \mathfrak{m}_{k-j}^{\rm iso}( {t;J_{k-j}})  \lesssim \mathfrak{m}_k^{\rm iso}( {t; J_{k}}) \quad \text{and} \quad 
				\s^{ \rm iso}_j(t; {J_{j}}) \, \s^{\rm iso}_{k-j}(t ; {J_{k-j}}) \lesssim \ell_t^{-1/2}\s^{\rm iso}_k(t ; {J_{k}})
			\end{split}
		\end{equation}
		for all disjoint decompositions $ {J_{k}} = {J_{j}} \dot{\cup} {J_{k-j}}$.
		\item[(ii)] The mean size can be upper bounded by the standard deviation size as   
		\begin{equation}
			\label{eq:2iniso}
			\mathfrak{m}_k^{\rm iso} (t; {J_{k}})\lesssim \ell_t^{1/2} \,  \s_k^{ \rm iso}(t; {J_{k}}) \,.
		\end{equation}
		\item[(iii)] The standard deviation size satisfies the doubling inequality
		\begin{equation}
			\label{eq:3iniso}
				\s_{2k}^{\rm iso}(t; J_{k} \cup J_{k})
			\lesssim \, \ell_t^{1/2} \, \left(\s_k^{ \rm iso}(t;{J_{k}})\right)^2\,,
		\end{equation}
		where $J_k \cup J_k$ denotes the union of the multiset $J_k$ with itself. 
		\item[(iv)] Monotonicity  in time: for any $0 \le s \le t$ and $\alpha \in [0,1/2]$, we have
		\begin{equation}
			\label{eq:4iniso}
			\mathfrak{m}_k^{\rm iso}(s;J_k) \lesssim \mathfrak{m}_k^{\rm iso}(t;J_k) \qquad \text{and} \qquad \s_k^{\rm iso}(s;J_k) \, \ell_s^{\alpha}\lesssim \s_k^{\rm iso}(t;J_k) \, \ell_t^{\alpha}\,. 
		\end{equation}
	\end{itemize}
\end{lemma}

Moreover, we have the following relations among $\m^{\rm iso}/\s^{\rm iso}$ and $\m/\s$ from \eqref{eq:sizes}, whose proofs 
are again  immediate from \eqref{eq:plSchattenrel} and hence omitted. 

\begin{lemma}[$\m^{\rm iso}/\s^{\rm iso}$-$\m/\s$-relations] \label{lem:m-relnoIMaviso}
	Let $k \ge 1$, consider a multiset $ \bm{A}_{J_{k}}$ of traceless matrices as above and fix $\ell > 0$.  Then, the mean sizes $\mathfrak{m}^{\rm iso}_k $ and $\mathfrak{m}_k $, 
	and the standard deviation sizes $\s^{\rm iso}_k $ and $\s_k $,  
	satisfy the following  inequalities. 
	\begin{itemize}
		\item[(i)] The \emph{average sizes} $\m/\s$ are bounded by the isotropic sizes $\m^{\rm iso}/\s^{\rm iso}$: 
		\begin{equation}
			\label{eq:1inaviso}
			\begin{split}
				\m_k(\ell;J_k)  \lesssim 	\m_k^{\rm iso}(\ell;J_k) \qquad \text{and} \qquad \s_k(\ell;J_k)  \lesssim \ell^{1/2} \, 	\s_k^{ \rm iso}(\ell;J_k) \,. 
			\end{split}
		\end{equation}
		\item[(ii)] The isotropic standard deviation size can be upper bounded by the average standard deviation size 
		\begin{equation}
			\label{eq:2inaviso}
			\s_k^{ \rm iso} (\ell;J_k) \lesssim \, N \, \big[1  + (N\ell)^{-1/2}\big]\,  \s_k(\ell;J_k) \,.
		\end{equation}
	\end{itemize}
\end{lemma}
~\\[2mm]
\noindent\textbf{Stochastic control quantities.} Consider deterministic vectors $\bm x, \bm y \in \C^N$ with $\Vert \bm x \Vert, \Vert \bm y \Vert \lesssim 1$, traceless matrices $A_1,\dots,A_k$, and for $k\ge 1$ define, analogously to \eqref{eq:defphi}--\eqref{eq:defpsi},
\begin{equation}
	\label{eq:defpsiiso}
	\begin{split}
	\Psi_k^{\rm iso}(t) &= \Psi_k^{\rm iso}(t; \bm x, \bm y; \bm z_{[1,k+1]}; \bm A_{[1,k]}):= \frac{N^{1/2}}{\s_k^{\rm iso}(\ell_t ; \bm A_{[1,k]})(N\ell_t)^{1/4}} \big| (G_{[1,k+1],t}-M_{[1,k+1],t})_{{\bm x}{\bm y}} \big|\,,  \\
	\Phi_k^{\rm iso}(t) &= \Phi_k^{\rm iso}(t; \bm x, \bm y; \bm z_{[1,k+1]}; \bm A_{[1,k]}):= \frac{N^{1/2}}{\s_k^{\rm iso}(\ell_t ; \bm A_{[1,k]})} \big| (G_{[1,k+1],t}-M_{[1,k+1],t})_{{\bm x}{\bm y}} \big|\,, 
	\end{split}
\end{equation}
where we denoted the multisets of deterministic matrices by $\bm A_{[1,k]}$ and $\bm z_{[1,k+1]} = (z_1, ... , z_{k+1})$ with $z_i = z_{i,0}$ (initial condition), respectively.
Note that by \eqref{eq:singlegllaw}, using the convention $\s_0^{\rm iso}(t):=\ell_t^{-1/2}$, we have $\Phi_0(t)\prec 1$, and that by \eqref{eq:mainresultinISO} it follows
\begin{equation}
	\label{eq:inbisoneed}
	\Phi_k^{\rm iso}(0)\prec 1,
\end{equation}
for any $k\ge 1$. Similarly to the averaged case (see the two paragraphs around \eqref{eq:defphi}--\eqref{eq:defpsi}), also in the isotropic case we introduced the two quantities in \eqref{eq:defpsiiso} as we will first prove  $\Psi_k^{\rm iso}(t)=\Phi_k^{\rm iso}(t)/(N\ell_t)^{1/4}\prec 1$ using the master inequalities in Proposition~\ref{pro:masteriniso}, and then use this as an input to prove $\Phi_k^{\rm iso}(t)\prec 1$ in Proposition~\ref{pro:masterinisoII}. 

\begin{proposition}[Isotropic master inequalities]
	\label{pro:masteriniso}
	Fix $k\in \N$, assume that $\Psi_l^{\rm iso}(t)\prec \psi_l^{\rm iso}$, with $\psi_0^{\rm iso}:=1$ and $\Phi_l(t)\prec 1$ (with $\Phi_l(t)$ from \eqref{eq:defpsi}), for any $0\le l \le k+ \mathbf{1}(k \, \mathrm{odd})$ uniformly in $t\in [0,T]$. Then
	\begin{equation}
		\label{eq:masterisonew}
		\Psi_k^{\rm iso}(t)\prec 1+\frac{\psi_k^{\rm iso}+\psi_{k+\bm1(k \,  \mathrm{odd})}^{\rm iso}}{(N\ell_T)^{1/8}} +\sum_{l=1}^{k-1}\psi_l^{\rm iso}
	\end{equation}
	uniformly in $t\in [0,T]$.
\end{proposition}
Analogously to Section \ref{sec:optA}, using Proposition~\ref{pro:masteriniso} we show that $\Psi_k^{\rm iso}(t)\prec 1$ and then use this information as input to conclude $\Phi_k^{\rm iso}(t)\prec 1$:
\begin{proposition}
	\label{pro:masterinisoII}
	Fix $k\in \N$, assume that $\Phi_l^{\rm iso}(t)\prec1$, for $0\le l \le k-2$, $\Phi_l(t)\prec 1$ for $0\le l \le k$ (with $\Phi_l(t)$ from \eqref{eq:defphi}), and $\Phi_k(t)\prec (N\ell_t)^{1/4}$ for $l\le 2k$, uniformly in $t\in [0,T]$. Then
	\begin{equation}
		\label{eq:masterisonewll}
		\Phi_k^{\rm iso}(t)\prec 1
	\end{equation}
	uniformly in $t\in [0,T]$.
\end{proposition}

\noindent \textbf{Closing the hierarchy.} We start showing that \eqref{eq:masterisonew}, using \emph{iteration} as in Lemma~\ref{lem:iteration}, in fact implies $\Psi_k(t)\prec 1$. We start with $k=1,2$. By \eqref{eq:masterisonew}, we have
\begin{equation}
\label{eq:stepiso1}
\Psi_1^{\rm iso}(t)\prec 1+\frac{\psi_1^{\rm iso}+\psi_2^{\rm iso}}{(N\ell_T)^{1/8}}, \qquad\quad \Psi_2^{\rm iso}(t)\prec 1+\frac{\psi_2^{\rm iso}}{(N\ell_T)^{1/8}}+\psi_1^{\rm iso}.
\end{equation}
Then, by iteration, we obtain
\begin{equation}
\label{eq:stepiso2}
\Psi_1^{\rm iso}(t)\prec 1+\frac{\psi_2^{\rm iso}}{(N\ell_T)^{1/8}}, \qquad\quad \Psi_2^{\rm iso}(t)\prec 1+\psi_1^{\rm iso}.
\end{equation}
Finally, plugging the first inequality into the second one and using iteration once again we obtain 
$\Psi_1^{\rm iso}(t)+\Psi_2^{\rm iso}(t)\prec 1$.

Next, we proceed by induction. Fix an even $k\in \N$, and assume that $\Psi_l^{\rm iso}(t)\prec 1$ for $l\le k-2$, then by \eqref{eq:masterisonew} we have
\begin{equation}
\label{eq:stepiso3}
\Psi_{k-1}^{\rm iso}(t)\prec 1+\frac{\psi_{k-1}^{\rm iso}+\psi_k^{\rm iso}}{(N\ell_T)^{1/8}}, \qquad\quad \Psi_k^{\rm iso}(t)\prec 1+\frac{\psi_k^{\rm iso}}{(N\ell_T)^{1/8}}+\psi_{k-1}^{\rm iso}.
\end{equation}
Proceeding exactly as in \eqref{eq:stepiso1}--\eqref{eq:stepiso2}, by \eqref{eq:stepiso3}, we conclude that $\Psi_l^{\rm iso}(t)\prec 1$ for any $l\le k$. Finally, using that as a consequence of $(N\ell_t)^{1/4}\Phi_k^{\rm iso}(t)=\Psi_l^{\rm iso}(t)\prec 1$  the hypothesis of Proposition~\ref{pro:masterinisoII} are satisfied, we conclude that $\Phi_k^{\rm iso}(t)\prec 1$ and so the proof of Theorem~\ref{thm:isolaw}.
\qed 

\subsubsection{Isotropic master inequalities: Proofs of Propositions~\ref{pro:masteriniso}--\ref{pro:masterinisoII}} \label{subsubsec:isomasred}

\begin{proof}[Proof of Proposition~\ref{pro:masteriniso}]
	
	Note that the second term in the first line of \eqref{eq:flowkaiso} can be incorporated into the lhs. by differentiating $e^{-kt/2}(G_{[1,k+1,t]}-M_{[1,k+1,t]})_{{\bm x}{\bm y}}$. The exponential factor $e^{kt/2}\sim 1$ is irrelevant, we thus neglect this term from the analysis. In the following we will often use the simple bound \eqref{etaint} even if not stated explicitly. Additionally, every time that two resolvents get next to each other in a chain
	 we use the integral representation \eqref{eq:intrep} to reduce their number by one at the price of an additional $1/\eta$--factor (see e.g. \eqref{eq:shortenchain}).
	
	We now start with the computation of the quadratic variation of the martingale term in \eqref{eq:flowkaiso}:
	\begin{equation}
	\label{eq:quadvarexpiso}
		\frac{1}{N}\sum_{i=1}^{k+1} \big[G_{[1,i],t}(G_{[1,i],t})^*\big]_{{\bm x}{\bm x}}\big[(G_{[i,k+1],t})^*G_{[i,k+1],t}\big]_{{\bm y}{\bm y}} \, \dd t.
	\end{equation}
	 Similarly to the averaged case (see \eqref{eq:quadvarexp}--\eqref{eq:prob2}), to estimate the quadratic variation of 
	 the martingale term in \eqref{eq:flowkaiso} we rely on the following \emph{reduction inequality} for $k\in \N$ even and $j\le k$ (here for simplicity we drop the indices of $G$'s and $A$'s):
		\begin{equation}
	\label{eq:rediniso}
	\big|(G_{[1,k+j+1]})_{{\bm x}{\bm y}}\big|\lesssim \sqrt{N}\langle |G|A(GA)^{j-1}|G|(AG^*)^{j-1}A\rangle^{1/2}\prod_{{\bm v}\in\{{\bm x},{\bm y}\}}[(GA)^{k/2}|G|(AG^*)^{k/2}]_{{\bm v}{\bm v}}^{1/2}.
	\end{equation}
	The proof of \eqref{eq:rediniso} is postponed to the end of this section. Recall the conventions $\psi_0^{\rm iso}=1$ 
	and $\mathfrak{s}_0^{\rm iso}(t)=\ell_t^{-1/2}$, 
	then using \eqref{eq:rediniso} for $j=k-2i+2$, $i\le k/2$ for even $k$ (and $j=k-2i+1$, $i\le (k+1)/2$ for odd $k$) and then \eqref{eq:2iniso}, \eqref{eq:2in}, to bound (recall \eqref{eq:absGintrep} in Footnote \ref{ftn:absval})
	\begin{equation}
	\label{eq:usefred}
	\begin{split}
	\big[(G_{[i,k+1],s})^*G_{[i,k+1],s}\big]_{{\bm y}{\bm y}}&\prec\sqrt{N}\left(\mathfrak{m}_k^{\rm iso}(s)+\frac{\mathfrak{s}_k^{\rm iso}(s)(N\ell_s)^{1/4}}{\sqrt{N}}\right)\left(\ell_s\mathfrak{m}_{2(k-2i+2)}(s)+\frac{\mathfrak{s}_{2(k-2i+2)}(s)}{N}\right)^{1/2} \\
	&\lesssim \mathfrak{s}_k^{\rm iso}(s)\sqrt{\mathfrak{s}_{2(k-2i+2)}(s)}\ell_s^{1/2}\big(\sqrt{N\ell_t}+(N\ell_t)^{1/4}\psi_k^{\rm iso}\big),
	\end{split}
	\end{equation} 
we estimate each term of \eqref{eq:quadvarexpiso} by
	\begin{equation}
		\begin{split}
			\label{eq:estquadvariso}
			&\frac{\sqrt{N}}{\mathfrak{s}_k^{\rm iso}(t)(N\ell_t)^{1/4}}\left(\frac{1}{N}\sum_{i=1}^{k+1}\int_0^t \big[G_{[1,i],s}(G_{[1,i],s})^*\big]_{{\bm x}{\bm x}}\big[(G_{[i,k+1],s})^*G_{[i,k+1],s}\big]_{{\bm y}{\bm y}} \, \dd s\right)^{1/2} \\
			&\prec\frac{\sqrt{N}}{\mathfrak{s}_k^{\rm iso}(t)(N\ell_t)^{1/4}}\Bigg(\sum_{i=1}^{\lceil k/2\rceil}\int_0^t\frac{\mathfrak{s}_k^{\rm iso}(s)}{N\eta_s^2} \left(\mathfrak{m}_{2(i-1)}^{\rm iso}(s)+\frac{\psi_{2(i-1)}^{\rm iso}\mathfrak{s}_{2(i-1)}^{\rm iso}(s)(N\ell_s)^{1/4}}{\sqrt{N}}\right)  \\
			&\qquad\qquad\qquad\qquad\qquad\qquad\qquad\qquad\qquad\times\big(\sqrt{N\ell_s}+\psi_{k+\bm1(k\, \mathrm{odd})}^{\rm iso}(N\ell_s)^{1/4}\big)\sqrt{\ell_s\mathfrak{s}_{2(k-i+1)}(s)} \, \dd s\Bigg)^{1/2} \\
			&\lesssim\frac{\sqrt{N}}{\mathfrak{s}_k^{\rm iso}(t)(N\ell_t)^{1/4}}\left(\sum_{i=1}^{\lceil k/2\rceil}\int_0^t\frac{\mathfrak{s}_k^{\rm iso}(s)\mathfrak{s}_{k-i+1}^{\rm iso}(s)\mathfrak{s}_{i-1}^{\rm iso}(s)}{N\eta_s^2} \ell_s^{3/2} \left(1+\frac{\psi_{2(i-1)}^{\rm iso}}{(N\ell_s)^{1/4}}\right) \big(\sqrt{N\ell_s}+\psi_{k+\bm1(k\, \mathrm{odd})}^{\rm iso}(N\ell_s)^{1/4}\big) \, \dd s\right)^{1/2} \\
			&\lesssim 1+\sum_{i=1}^{\lceil k/2\rceil}\frac{\sqrt{\psi_{2(i-1)}^{\rm iso}}}{(N\ell_s)^{1/8}}+\frac{\sqrt{\psi_{k+\bm1(k\, \mathrm{odd})}^{\rm iso}}}{(N\ell_t)^{1/8}}+\sum_{i=1}^{\lceil k/2\rceil}\frac{\sqrt{\psi_{2(i-1)}^{\rm iso}\psi_{k+\bm1(k\, \mathrm{odd})}^{\rm iso}}}{(N\ell_t)^{1/4}}.
		\end{split}
	\end{equation}
Here in the second inequality \eqref{eq:2iniso}--\eqref{eq:3iniso}, \eqref{eq:1inaviso}, and in the last inequality we used \eqref{eq:1iniso}, \eqref{eq:4iniso}.

	In the following computations, when two $G$'s, with spectral parameters having imaginary parts of the same sign appear next to each other (i.e.~without a $A$ in between), we use the integral representation \eqref{eq:intrep} to reduce the number of $G$'s by one at the price of an additional $1/\eta$. If the imaginary parts of the spectral parameters have different signs, we use resolvent identity (see e.g. around \eqref{eq:shortenchain}--\eqref{eq:intrep} in the averaged case). For the terms in the second line of \eqref{eq:flowkaiso} we estimate
	\begin{equation}
	\label{eq:avisoest}
		\frac{\sqrt{N}}{\mathfrak{s}_k^{\rm iso}(t)(N\ell_t)^{1/4}}\int_0^t \langle {G}_{[i,j],s}-M_{[i,j],s}\rangle ({M}_{[1,i]\cup [j,k+1],s})_{{\bm x}{\bm y}}\,\dd s\prec\frac{\sqrt{N}}{\mathfrak{s}_k^{\rm iso}(t)(N\ell_t)^{1/4}}\int_0^t \frac{\mathfrak{s}_{j-i}(s)}{N\eta_s} \frac{\mathfrak{m}_{k-j+i}^{\rm iso}(s)}{\eta_s}\,\dd s\lesssim \frac{1}{(N\ell_t)^{3/4}},
	\end{equation}
	where we used \eqref{eq:mboundiso}--\eqref{eq:2iniso}, the second inequality in \eqref{eq:1inaviso}, and \eqref{eq:4iniso}.
	
	For the terms in the third line of \eqref{eq:flowkaiso} we estimate
	\begin{equation}
		\label{eq:intermstepiso}
		\begin{split}
		&\frac{\sqrt{N}}{\mathfrak{s}_k^{\rm iso}(t)(N\ell_t)^{1/4}}\int_0^t \langle {M}_{[i,j],s}\rangle ({G}_{[1,i]\cup [j,k+1],s}-M_{[1,i]\cup [j,k+1],s})_{{\bm x}{\bm y}}\,\dd s \\
		&\qquad\qquad\quad\prec \frac{\sqrt{N}}{\mathfrak{s}_k^{\rm iso}(t)(N\ell_t)^{1/4}}\int_0^t \frac{\mathfrak{m}_{j-i}(s)}{\eta_s} \frac{\mathfrak{s}_{k-j+i}^{\rm iso}(s)\psi_{k-j+i}^{\rm iso}(N\ell_s)^{1/4}}{\sqrt{N}\eta_s} \,\dd s\lesssim \psi_{k-j+i}^{\rm iso},
		\end{split}
	\end{equation}
	where in the last inequality we used the first inequality of \eqref{eq:1inaviso}, \eqref{eq:1iniso}--\eqref{eq:2iniso}.
	
	For the terms in the fourth line of \eqref{eq:flowkaiso} we estimate
	\begin{equation}
		\begin{split}
			&\frac{\sqrt{N}}{\mathfrak{s}_k^{\rm iso}(t)(N\ell_t)^{1/4}}\int_0^t \langle {G}_{[i,j],s}-M_{[i,j],s}\rangle ({G}_{[1,i]\cup [j,k+1],s}-M_{[1,i]\cup [j,k+1],s})_{{\bm x}{\bm y}}\,\dd s \\
			&\qquad\qquad\quad\prec \frac{\sqrt{N}}{\mathfrak{s}_k^{\rm iso}(t)(N\ell_t)^{1/4}}\int_0^t \frac{\mathfrak{s}_{j-i}(s)}{N\eta_s} \frac{\mathfrak{s}_{k-j+i}^{\rm iso}(s) \psi_{k-j+i}^{\rm iso}(N\ell_s)^{1/4}}{\sqrt{N}\eta_s} \,\dd s\lesssim \frac{1}{N\ell_t} \psi_{k-j+i}^{\rm iso},
		\end{split}
	\end{equation}
	where in the last step we used the second inequality of \eqref{eq:1inaviso} and the second inequality of \eqref{eq:1iniso}.
	
	Finally, for the term in the last line of \eqref{eq:flowkaiso} we estimate
	\begin{equation}
		\begin{split}
			\label{eq:estlastterm}
			\frac{\sqrt{N}}{\mathfrak{s}_k^{\rm iso}(t)(N\ell_t)^{1/4}}\int_0^t \langle G_{i,s}-m_{i,s}\rangle  ({G}^{(i)}_{[1,k+1],s})_{{\bm x}{\bm y}}\,\dd s &\prec \frac{\sqrt{N}}{\mathfrak{s}_k^{\rm iso}(t)(N\ell_t)^{1/4}} \int_0^t \frac{1}{N\eta_s}\left(\frac{\mathfrak{m}_k^{\rm iso}(s)}{\eta_s}+\frac{\mathfrak{s}_k^{\rm iso}(s)\psi_k^{\rm iso}(N\ell_s)^{1/4}}{\sqrt{N}\eta_s} \right)\,\dd s \\
			&\lesssim \frac{1}{(N\ell_t)^{3/4}}+\frac{\psi_k^{\rm iso}}{N\ell_t},
		\end{split}
	\end{equation}
	where we used \eqref{eq:mboundiso} and \eqref{eq:2iniso} together with \eqref{eq:2iniso}. Combining \eqref{eq:estquadvariso}--\eqref{eq:estlastterm}, recalling that $\langle M_{[i,i+1],s}\rangle=0$ in \eqref{eq:intermstepiso}, using \eqref{eq:inbisoneed} to bound $\Psi_k^{\rm iso}(0)\prec 1$ and $N\ell_T\ge1$, $\psi_l^{\rm iso}\ge 1$, $\ell_t\gtrsim \ell_T$ , we conclude \eqref{eq:masterisonew}.
\end{proof}

\begin{proof}[Proof of Proposition~\ref{pro:masterinisoII}]

The proof of this proposition is analogous to the proof of Proposition~\ref{pro:masterinisoII}. All the terms except for the 
martingale one are estimated as in the proof of Proposition~\ref{pro:masterinisoII} after multiplying each line by $(N\ell_t)^{1/4}$ and setting $\psi_l^{\rm iso}=1$. We conclude the proof pointing out that the only difference is in the estimate of the quadratic variation of the martingale term, i.e. \eqref{eq:estquadvariso} has to be replaced by
\begin{equation}
\begin{split}
&\frac{\sqrt{N}}{\mathfrak{s}_k^{\rm iso}(t)}\left(\frac{1}{N}\sum_{i=1}^{k+1}\int_0^t \big[G_{[1,i],s}(G_{[1,i],s})^*\big]_{{\bm x}{\bm x}}\big[(G_{[i,k+1],s})^*G_{[i,k+1],s}\big]_{{\bm y}{\bm y}} \, \dd s\right)^{1/2} \\
&\qquad\qquad\qquad\qquad\qquad\qquad\qquad\quad\prec \frac{\sqrt{N}}{\mathfrak{s}_k^{\rm iso}(t)}\left(\int_0^t \frac{\ell_s\mathfrak{s}_{2(i-1)}^{\rm iso}(s)\mathfrak{s}_{2k-2(i-1)}^{\rm iso}(s) }{N\eta_s^2} \, \dd s\right)^{1/2} \lesssim 1,
\end{split}
\end{equation}
where in the first inequality we used the definition \eqref{eq:defpsiiso} together with \eqref{eq:mboundiso}, \eqref{eq:2iniso}, and in the second inequality we used \eqref{eq:1iniso}, \eqref{eq:3iniso}--\eqref{eq:4iniso}.
\end{proof}

\begin{proof}[Proof of \eqref{eq:rediniso}]
By spectral decomposition we estimate
\begin{equation}
\begin{split}
\big|(Q_1G(z)Q_2G(w)Q_3)_{{\bm x}{\bm y}}\big|&=\left|\sum_{ij}\frac{\langle {\bm x}, Q_1{\bm u}_i\rangle\langle {\bm u}_i, Q_2{\bm u}_j\rangle\langle {\bm u}_j, Q_3{\bm u}_j\rangle}{(\lambda_i-z)(\lambda_j-w)}\right| \\
&\lesssim \sqrt{N} (Q_1|G(z)|Q_1^*)_{{\bm x}{\bm x}}(Q_3^*|G(w)|Q_3)_{{\bm y}{\bm y}}\langle |G(z)|Q_2|G(w)|Q_2^*\rangle^{1/2},
\end{split}
\end{equation}
where in the last inequality we used Schwarz inequality to separate $Q_2$ from $Q_1,Q_3$. Choosing $z=z_{k/2+1}$, $w=z_{k/2+j+1}$, and
\[
Q_1=G_1A_1\dots A_{k/2}, \qquad Q_2=A_{k/2+1}G_{k/2+2}\dots A_{k/2+j}, \qquad Q_3=A_{k/2+j+1}G_{k/2+j+2}\dots G_{k+1},
\]
this concludes the proof.
\end{proof}

\end{document}